\documentclass[a4paper,12pt]{article}
\usepackage[a4paper,left=2.5cm, right=2.5cm, top=2.5cm, bottom=2.5cm]{geometry}
\usepackage[utf8]{inputenc}   	
\usepackage[T1,TS1]{fontenc}        

\usepackage{rotating}
\usepackage{threeparttable}
\usepackage{bbm}
\usepackage{lineno}
\usepackage{adjustbox}
\usepackage[section]{placeins}

\usepackage{setspace} 
\usepackage{color}              
\usepackage{amsmath,amsthm,amsfonts,amssymb,amscd,mathrsfs}
\usepackage{bm}						
\usepackage{bbm}	
\usepackage[ngerman, english]{babel}     %
\usepackage{ae}                 %
\usepackage{authblk}
\usepackage{graphicx}           
\usepackage{braket}
\usepackage{longtable}          
\usepackage{multicol}
\usepackage{array}
\usepackage[inline]{enumitem}
\usepackage{multirow}           
\usepackage[hyphens]{url}
\usepackage{pdfpages}
\usepackage{epsfig}
\usepackage{fancybox}
\usepackage[labelfont={bf}, hypcap=false,justification=centerfirst]{caption}
\usepackage{bold-extra}
\usepackage{xcolor}
\usepackage{fancyvrb}
\usepackage{arydshln}
\usepackage{tcolorbox}
\tcbuselibrary{breakable}

\usepackage{lscape}
\usepackage{import}
\usepackage{float}
\usepackage{bigstrut}
\usepackage[titletoc]{appendix}
\usepackage{todonotes}
\usepackage[babel]{csquotes}
\usepackage[babel]{csquotes}
\usepackage{datetime}
\newdateformat{deformat}{\THEDAY{.} \monthname[\THEMONTH] \THEYEAR}
\definecolor{tuebingendarkred}{RGB}{165,30,55}
\usepackage{alphalph}
\usepackage{pifont}
\usepackage{ae}
\usepackage{subcaption}
\usepackage{microtype}
\usepackage{placeins}
\usepackage{bm,xstring}
\usepackage{cleveref}
\usepackage{tabularx}
\usepackage[hang,bottom,stable]{footmisc}
\newcolumntype{d}{D{.}{.}{-1}}
\newcolumntype{s}{>{\hsize=1\hsize}X}

\RecustomVerbatimCommand{\VerbatimInput}{VerbatimInput}%
{fontsize=\footnotesize,
 frame=lines,  
 framesep=2em, 
 rulecolor=\color{Gray},
 label=\fbox{\color{Black}data.txt},
 labelposition=topline,
 commandchars=\|\(\), 
 commentchar=*        
}

\usepackage{amsmath}
\usepackage{amssymb}
\usepackage{amsthm}

\newtheorem{proposition}{Proposition}

\newtheorem{remark}{Remark}

\newcommand{\seq}[1]{\left<#1\right>}

    \def\mb#1%
	{\IfSubStr{ABCDEFGHIJKLMNOPQRSTUVWXYZabcdefghijklmnopqrstuvwxyz}{#1}
	  {\mathbf{#1}}
	  {\bm{#1}}}
    
\DeclareMathOperator*{\Prob}{Prob}		

\def\mb#1%
{\IfSubStr{ABCDEFGHIJKLMNOPQRSTUVWXYZabcdefghijklmnopqrstuvwxyz}{#1}
    {\mathbf{#1}}
    {\bm{#1}}}

\makeatletter
\def\munderbar#1{\underline{\sbox\tw@{$#1$}\dp\tw@\z@\box\tw@}}
\let\@@pmod\pmod
\DeclareRobustCommand{\pmod}{\@ifstar\@pmods\@@pmod}
\def\@pmods#1{\mkern4mu({\operator@font mod}\mkern 6mu#1)}
\makeatother
\makeatletter 
\def\leqn{\tagsleft@true} 
\def\reqn{\tagsleft@false} 
\def\fleq{\@fleqntrue \let\mathindent\@mathmargin \@mathmargin=\leftmargini} 
\def\cneq{\@fleqnfalse} 
\makeatother 

\usepackage{pgf}
\usepackage{tikz}
\usetikzlibrary{arrows,shapes,automata,backgrounds,petri,calc,matrix,decorations.markings,decorations.pathreplacing,patterns}
\usepackage[titletoc]{appendix}

\newcolumntype{C}[1]{>{\centering\arraybackslash}m{#1}}
\newcolumntype{R}[1]{>{\raggedleft\arraybackslash}m{#1}}  

\usepackage{booktabs}
\usepackage{environ}
\usepackage{expl3}
\ExplSyntaxOn
\seq_new:N \l_DT_contents_seq
\seq_new:N \l_DT_contents_A_seq
\tl_new:N \l_DT_header_tl
\cs_generate_variant:Nn \seq_set_split:Nnn { NnV }
\NewEnviron{dbltbl}[2]
  {
    \seq_set_split:NnV \l_DT_contents_seq { \\ } \BODY

    \seq_pop_left:NN \l_DT_contents_seq \l_DT_header_tl

    \seq_clear:N \l_DT_contents_A_seq
    \prg_replicate:nn
      { \int_div_truncate:nn { \seq_count:N \l_DT_contents_seq } {2} }
      {
        \seq_pop:NN \l_DT_contents_seq \l_tmpa_tl
        \seq_put_right:NV \l_DT_contents_A_seq \l_tmpa_tl
      }

    \begin{tabular}{#1m{#2}m{#2}#1}
      \toprule
      \l_DT_header_tl &&& \l_DT_header_tl \\
      \midrule
      \seq_mapthread_function:NNN
        \l_DT_contents_A_seq
        \l_DT_contents_seq
        \DT_one_row:nn
      \bottomrule
    \end{tabular}
  }
\cs_new:Npn \DT_one_row:nn #1#2 { #1 &&& #2 \\ }
\ExplSyntaxOff

\usepackage{quoting}

\floatstyle{plain}
\newfloat{Example}{htbp}{loe}
\floatname{Example}{Example}
\usepackage{natbib}
\date{\today}
\begin{document}

\begin{titlepage}
\singlespace

\title{A Bayes-Factor-Guided Approach to Post-Double Selection with Bootstrapped Multiple Imputation}
 \author[1]{Johannes Bleher}
 \author[2]{Claudia Tarantola}

 \affil[1]{Department of Econometrics and Empirical Economics \& Computational Science Hub, University of Hohenheim}
 \affil[2]{Department of Economics, Management and Quantitative Methods, University of Milan}
\date{\today}

\maketitle
\thispagestyle{empty}

\vspace*{-0.5cm}

\begin{abstract}
\noindent
When variable selection methods are applied to bootstrapped and multiply imputed datasets, the set of selected variables typically varies across iterations. Aggregating results via the union rule can lead to overly dense models. We propose a sequential evidence aggregation procedure that models detection outcomes across perturbation iterations as Bernoulli trials and accumulates evidence for variable relevance through a likelihood-ratio process admitting an approximate Bayes-factor interpretation. The procedure provides both a variable inclusion criterion and a stopping rule that eliminates the need to fix the number of bootstrap -- imputation iterations ex ante. A Monte Carlo study across 126 scenarios and an empirical illustration demonstrate the method's performance relative to existing aggregation approaches.
\end{abstract}
\vspace{0.5cm}

\begin{tabular}{p{0.15\textwidth} p{0.7\textwidth}}
  \textit{Keywords:} &
  Bayes Factor,
  Bootstrap,
  Multiple Imputation,
  Post-Double Selection,
  Sequential Testing,
  Variable Selection \\
  \\
\textit{JEL:} & C11, C15, C21, C52\\
\end{tabular}
\end{titlepage}

\doublespacing

\section{Introduction}

Estimating the effect of a variable of interest on an outcome in the presence of many potential control variables is a common problem in empirical research. When the number of candidate covariates is large, variable selection methods such as LASSO are routinely used to construct parsimonious models. However, data-driven selection can distort estimation if relevant variables are omitted, particularly when they are associated with both the variable of interest and the outcome \citep{Hayashi2000}.

The post-double selection (PDS) approach of \citet{BelloniCH2014b,BelloniCH2014a} addresses this issue by performing variable selection in both the outcome and the variable-of-interest equations and combining the resulting sets of selected variables. This construction increases the likelihood that relevant controls are retained and has become a standard tool for estimation after selection in high-dimensional settings.

In empirical applications, variable selection is often combined with additional procedures that introduce further stochastic variability. Missing covariate data are commonly handled via multiple imputation (MI) \citep{Rubin1987,LittleR2019}, while sampling uncertainty is frequently addressed through bootstrap resampling \citep{Efron1979,EfronT1994}. When selection methods are applied repeatedly across such perturbed datasets, the set of selected variables typically varies across iterations. Standard aggregation rules, such as the union of all selected variables, may therefore lead to overly large models, while frequency-based rules may discard variables with weaker but persistent signals.

Although perfect variable selection is not required for reliable estimation after selection \citep{BelloniCH2014b}, models containing many irrelevant covariates are often undesirable in practice. Moreover, researchers are often interested not only in obtaining stable estimates but also in identifying a meaningful subset of relevant control variables. In the presence of missing data and resampling procedures, assessing variable relevance therefore becomes particularly challenging.

This paper proposes a sequential evidence-based aggregation procedure for variable-selection results generated by repeated bootstrap and multiple-imputation (BOOT-MI) replications that use PDS screening within each iteration. The procedure uses PDS to construct within-iteration candidate sets, but applies an additional detection and aggregation layer across BOOT-MI replications. We implement a sequential BOOT-MI procedure in which each iteration consists of bootstrapping the incomplete dataset, performing stochastic imputation, constructing a PDS candidate set via LASSO in the outcome and variable-of-interest equations, and selectively confirming outcome-selected variables in a final regression step.

Rather than aggregating selected variables using heuristic rules, we construct a likelihood-ratio–based evidence measure with a Bayes-factor interpretation from a working detection model with empirically calibrated null and alternative detection probabilities. We then accumulate evidence for variable relevance sequentially across iterations. The resulting evidence measure provides a decision criterion for variable inclusion and induces a natural stopping rule for the sequential BOOT-MI procedure.

Our contribution is threefold.
First, we cast the aggregation of variable-selection outcomes across BOOT-MI datasets as a sequential decision problem. Specifically, we model binary detection indicators as outcomes of a stochastic detection process with hypothesis-dependent probabilities, leading to a likelihood-ratio-type evidence accumulation process. Second, we derive a stopping rule and a variable inclusion criterion based on this evidence process, eliminating the need to fix the number of perturbation iterations ex ante. Third, we propose a pilot-based calibration strategy for the detection probabilities and examine its performance across a range of simulation scenarios.

The proposed procedure is not intended as a new version of PDS itself, but as a calibrated sequential evidence-aggregation framework built around repeated BOOT-MI implementations that use PDS candidate screening. It combines an asymmetric within-iteration detection rule, pilot-based calibration of detection probabilities, and sequential log-evidence accumulation across perturbations. The resulting method is designed to yield interpretable and stable inclusion decisions in settings where repeated perturbations induce substantial variability in selected models.

The remainder of the paper is structured as follows. Section~\ref{sec:literature} reviews the related literature. Section~\ref{sec:methodology} develops the proposed sequential evidence-aggregation framework. Section~\ref{sec:simulation} reports the Monte Carlo study, and Section~\ref{sec:empirical} provides an empirical illustration. Section~\ref{sec:conclusion} provides concluding remarks and discusses directions for future research.

\section{Related Literature}
\label{sec:literature}

This paper relates to four strands of the literature: variable selection with multiply imputed data, bootstrap and resampling approaches to model selection, Bayesian evidence and Bayes factors, and sequential testing procedures. Each strand addresses a component of the problem considered here, but the combination of missing data, resampling-based variability, and principled aggregation of variable-selection outcomes remains only partially explored.

\subsection{Variable Selection with Multiply Imputed Data}

Variable selection in the presence of missing covariate data has received increasing attention in the statistical literature. MI provides a principled framework for handling incomplete datasets by replacing missing values with draws from their predictive distribution \citep{Rubin1987,WoodWR2008}. By generating multiple completed datasets, MI allows standard complete-data methods to be applied while accounting for imputation uncertainty.

Combining variable selection with MI introduces additional complexity, since selection results may differ across imputed datasets. Several approaches have been proposed to address this issue. \citet{ChenW2013} introduce MI-LASSO, which applies a group LASSO penalty across imputations to encourage consistent selection. Alternative approaches rely on stacked or grouped penalized regression formulations \citep{DuB2020}, while comparative studies evaluate the performance of LASSO-based and Bayesian procedures under missing data \citep{BainterM2023}. Bayesian methods, such as stochastic search variable selection and reversible jump Markov chain Monte Carlo, have also been applied in this setting \citep{GeorgeM1997,OHaraS2009}.

A related framework is PDS \citep{BelloniCH2014b,BelloniCH2014a}, originally developed in a causal inference context. The method performs variable selection in both the outcome equation and an auxiliary equation for the variable of interest and combines the resulting sets of selected variables. While this construction increases the likelihood that relevant controls are retained, it is typically combined with simple aggregation rules, such as the union of selected variables, which may lead to overly large models when applied across multiple imputations or resampled datasets.

Overall, this strand of literature highlights the difficulty of achieving stable and interpretable variable selection in the presence of missing data and provides limited guidance on how to aggregate selection outcomes across multiple stochastic completions of the data.

\subsection{Bootstrap, Resampling, and Stability-Based Variable Selection}

Resampling methods are widely used to assess sampling variability in variable selection procedures. Bootstrap-based approaches generate perturbed versions of the data and apply selection methods repeatedly, thereby providing information on the stability of selected variables. Several studies combine bootstrap resampling with MI to jointly account for sampling and imputation uncertainty \citep{LongJ2015,MusoroTORB2014}.

A closely related approach is stability selection \citep{MeinshausenB2010}, which evaluates variable importance based on selection frequencies across subsamples. Stability selection provides theoretical guarantees on the expected number of false selections under suitable assumptions and can be applied as a wrapper around a wide range of selection methods. However, its aggregation mechanism is based on frequency thresholds, which treat all detections symmetrically and do not account for the sequential structure of the selection process.

Recent work has also explored machine learning approaches to variable selection under missing data and resampling. For example, \citet{GunnR2022} study procedures that combine cross-validation, sample splitting, and MI in predictive modeling settings. These approaches emphasize predictive performance and model stability, but they typically rely on heuristic aggregation rules.

Taken together, resampling-based approaches highlight the intrinsic variability of variable selection under data perturbations. However, they generally rely on aggregation rules such as frequency thresholds or union rules, which may either include too many irrelevant variables or exclude variables with weaker but persistent signals. In contrast, the approach proposed in this paper aggregates selection outcomes through a sequential evidence-accumulation mechanism that accounts for the asymmetric informational content of detections and non-detections.

\subsection{Evidence Process and Bayes Factor Interpretation}

Complementary to resampling-based approaches, the literature has developed evidence-based frameworks for model comparison and variable selection. Bayes factors provide a formal mechanism for quantifying evidence in favor of competing hypotheses or models \citep{KassR1995}. Unlike classical hypothesis testing based on $p$-values, they allow direct comparison between alternative explanations and quantify evidence for both the null and the alternative.

The approach proposed in this paper differs from a fully Bayesian variable selection method. Rather than specifying a complete probabilistic model over parameters and model space, we construct a likelihood-ratio--based evidence measure with a Bayes-factor interpretation based on a working detection model with empirically calibrated parameters. This strategy is related to a broader literature on approximate Bayes factors, including the Schwarz criterion \citep{Schwarz1978} and calibrated Bayes factors that reinterpret frequentist test statistics in evidential terms \citep{Sellke2001}.

Unlike these approaches, which operate on model likelihoods or test statistics, our method is based on binary detection outcomes generated by repeated perturbation and selection. This perspective yields a tractable and interpretable evidence measure tailored to the aggregation of variable-selection results across BOOT-MI datasets.

\subsection{Sequential Testing and Evidence Accumulation}

Sequential hypothesis testing, introduced by \citet{Wald1945}, provides a framework for accumulating evidence through repeated observations and terminating data collection once sufficient evidence has been reached. The Sequential Probability Ratio Test (SPRT) is based on cumulative likelihood ratios and admits both frequentist and Bayesian interpretations. In particular, SPRT decision boundaries can be related to Bayes-factor thresholds under suitable prior specifications \citep{EdwardsLJ1963}.

The approach proposed in this paper draws on this connection by interpreting the cumulative likelihood ratio of the detection model as a sequential evidence measure with a Bayes-factor interpretation. However, unlike classical sequential testing settings, the observations in our framework are not independent draws from a sampling distribution but rather detection outcomes generated by stochastic perturbations of a fixed dataset. The resulting evidence accumulation mechanism is therefore motivated by, but not formally equivalent to, the SPRT framework.

\medskip

Despite these connections, previous works provide limited guidance on how to apply sequential evidence accumulation to the aggregation of variable-selection outcomes across BOOT-MI datasets. The present paper addresses this gap by developing a sequential evidence-based aggregation procedure tailored to this setting.

\section{Methodology}
\label{sec:methodology}

The methodology uses perturbation-based PDS screening as a within-iteration device and adds a detection model together with a sequential evidence-aggregation rule across perturbations.

\subsection{Perturbation-Based Variable Selection Framework}
\label{sec:perturbation}

We consider estimation in a sparse linear model with a large set of potential control variables, some of which may contain missing values. Let $Y$ denote the outcome variable, $D$ the variable of interest, and $X = (X_1,\dots,X_p)$ a vector of candidate covariates. Our parameter of interest is the coefficient $\alpha$ in the partially linear regression model
\begin{equation}
Y = D\alpha + X'\beta + \varepsilon.
\end{equation}

Following PDS \citep{BelloniCH2014b,BelloniCH2014a}, variable selection is based on two sparse regression problems corresponding to the outcome equation and an auxiliary equation for $D$. In particular, we consider LASSO regressions of $Y$ on $(D, X)$ and of $D$ on $X$, that is,
\begin{align}
\text{Outcome equation:} \quad & Y = D\alpha + X'\beta + \varepsilon, \\
\text{Auxiliary equation for } D: \quad & D = X'\gamma + \nu.
\end{align}
The outcome equation identifies variables predictive of $Y$, while the equation for $D$ identifies variables predictive of the variable of interest. In the standard PDS approach, the selected control set within a given dataset is formed by the union rule
\begin{equation}
S_t = S_{Y,t} \cup S_{D,t},
\end{equation}
where $S_{Y,t}$ and $S_{D,t}$ denote the sets of variables selected in the outcome equation and the auxiliary equation for $D$, respectively, in perturbation iteration $t$.

In each perturbation iteration, we generate a single completed dataset via stochastic imputation. Within the perturbation stage, each bootstrap draw is completed by one stochastic imputation and then passed to the selection/detection step. Thus, the repeated perturbation loop does not pool over multiple imputations within a given iteration. Instead, imputation uncertainty enters the selection stage through repeated stochastic completions across perturbations. Classical multiple-imputation pooling is used only after the final variable set has been determined, when the parameter of interest is re-estimated across $M$ completed datasets and combined using Rubin's rules \citep{Rubin1987}. Specifically, if $Q^{(m)}$ denotes the estimate of the parameter of interest and $U^{(m)}$ its complete-data variance estimate in imputed dataset $m=1,\dots,M$, the pooled point estimate is
\[
\bar{Q} = \frac{1}{M}\sum_{m=1}^M Q^{(m)}.
\]
The average within-imputation variance is
\[
\bar{U} = \frac{1}{M}\sum_{m=1}^M U^{(m)},
\]
and the between-imputation variance is
\[
B = \frac{1}{M-1}\sum_{m=1}^M \left(Q^{(m)}-\bar{Q}\right)^2.
\]
The total variance is then given by
\[
T = \bar{U} + \left(1+\frac{1}{M}\right)B.
\]
Hence, Rubin's rules combine estimation uncertainty within each imputed dataset and additional uncertainty arising from variation across imputations.

Missing covariate data are handled using MI, while sampling uncertainty is addressed through bootstrap resampling. In this paper we focus on the BOOT-MI approach, in which bootstrap resampling is performed before imputation. Let $t = 1,\dots,T$ index perturbation iterations. In each iteration, we proceed as follows:
\begin{enumerate}
\item Draw a bootstrap sample from the incomplete dataset.
\item Apply a stochastic imputation procedure to obtain a completed dataset.
\item Perform LASSO-based PDS on the completed dataset.
\end{enumerate}

Thus, each perturbation iteration produces a selected set $S_t$ of controls. Figure~\ref{fig:procedure} illustrates the workflow of the perturbation-based variable selection procedure.

Because both bootstrap sampling and stochastic imputation introduce randomness, the selected sets $S_t$ vary across perturbation iterations. While each individual set is typically sparse, the union across iterations,
\[
S_{\mathrm{union}} = \bigcup_{t=1}^T S_t,
\]
can become large when selection is unstable across perturbation iterations, thereby producing overly inclusive aggregated control sets.

To summarize variable-specific outcomes, we define for each variable $j$ and iteration $t$ a binary detection indicator $Z_{jt}$ (formalized in Section~\ref{sec:detection}). The sequence $\{Z_{jt}\}_{t=1}^{T}$ records whether variable $j$ is detected in each perturbed dataset and forms the basis for aggregation.

Rather than applying fixed aggregation rules such as unions or frequency thresholds, we interpret $\{Z_{jt}\}$ as realizations of a stochastic detection mechanism and aggregate them through a sequential evidence process.

\begin{landscape}
\begin{figure}
\centering
\caption{Procedure Overview for Bootstrapped Multiply Imputed Variable Selection}
\begin{minipage}{\linewidth}
\footnotesize
The figure illustrates how datasets containing missing values (gray shading) are bootstrapped and subsequently imputed. Variable selection is depicted by green shaded columns. The resulting selected variables can subsequently be aggregated using the proposed evidence measure.
\end{minipage}

\includegraphics[width=0.9\linewidth]{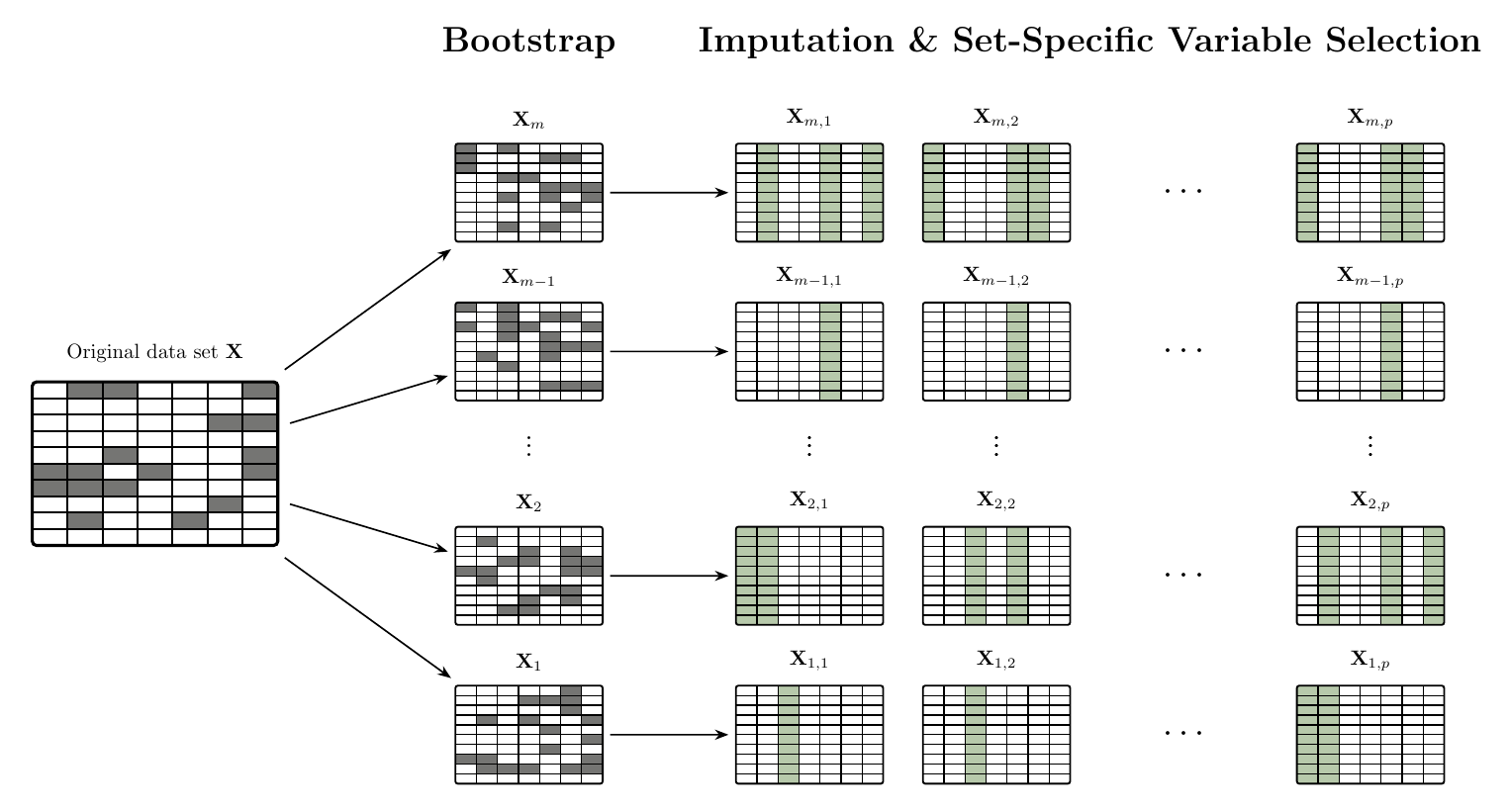}

\label{fig:procedure}
\end{figure}
\end{landscape}

\subsection{Model for Variable Relevance}
\label{sec:detection}

The perturbation-based procedure produces, for each variable $j$, a sequence of binary indicators $\{Z_{jt}\}_{t=1}^{T}$ summarizing variable-specific detection outcomes across perturbation iterations.

In contrast to pure selection indicators, we define a two-stage detection event that combines PDS screening with outcome-model confirmation. Specifically, let $S_{Y,t}$ denote the set of variables selected by LASSO in the outcome equation and let $S_{D,t}$ denote the set of variables selected in the auxiliary equation for $D$ at perturbation iteration $t$. We form the PDS candidate set
\[
S_t = S_{Y,t} \cup S_{D,t},
\]
which corresponds to the standard PDS union rule applied within each perturbation iteration.

To aggregate variable-selection outcomes across perturbation iterations, we require a variable-specific representation of the selection history. We therefore map the within-iteration PDS candidate set into a binary detection indicator for each variable and iteration. The next paragraph defines this detection rule, clarifies how it relates to the standard PDS inclusion rule, and then introduces the working probabilistic model that underlies sequential evidence accumulation.

\paragraph{Asymmetric detection rule.}
For each variable $j$ and iteration $t$, we define the detection indicator
\begin{align*}
Z_{jt} &= \mathbbm{1}\left\{ j \in S_{D,t} \;\lor\; \right. \\&\hspace{1.5cm} \left.\bigl( j \in S_{Y,t} \;\land\; \text{the coefficient of } X_j \text{ in } Y \sim D + X_{S_t} \text{ is significant at level } \alpha \bigr) \right\}.
\end{align*}

Thus, variables selected in the auxiliary equation ($j \in S_{D,t}$) are retained unconditionally, while variables selected only in the outcome equation ($j \in S_{Y,t}$) must additionally pass a significance test in the post-selection outcome regression.

This asymmetric rule reflects the structure of the PDS framework: variables predictive of $D$ are retained automatically, whereas variables selected only in the outcome equation are filtered more strictly to reduce overfitting and instability across perturbations. The resulting procedure should be interpreted as a detection rule built on the PDS candidate set rather than as the standard PDS inclusion rule. Accordingly, the formal inferential guarantees of standard PDS do not automatically carry over to the full procedure considered here. Because this detection rule is not identical to the standard PDS union rule, it is useful to state its relation to PDS explicitly.

\paragraph{Relation to PDS.}
The proposed detection rule departs from the standard PDS inclusion rule by applying an additional confirmation step to variables selected only in the outcome equation. In particular, defining
\[
S_t^{\mathrm{ASYM}} = \{ j : Z_{jt} = 1 \},
\]
we have
\[
S_t^{\mathrm{ASYM}} \subseteq S_t = S_{Y,t} \cup S_{D,t}.
\]
That is, the procedure can be interpreted as a data-dependent pruning of the PDS candidate set.

While the PDS framework provides formal guarantees under suitable conditions, these guarantees do not carry over to the modified rule. In particular, the procedure may omit weak but relevant controls, which can affect estimation of the coefficient on the variable of interest. The magnitude of this risk is therefore assessed empirically in Section~\ref{sec:simulation}.

Having clarified how the detection indicator is constructed and how it departs from standard PDS, we now use the resulting sequence $\{Z_{jt}\}$ as the basic input for the sequential evidence framework.

\paragraph{Detection model.}
For the purpose of sequential aggregation, we reduce the variable-specific decision problem to a binary distinction between relevance and irrelevance. We consider the hypotheses
\[
H_0: \text{Variable } j \text{ is irrelevant}, \qquad
H_1: \text{Variable } j \text{ is relevant}.
\]

Under this formulation, the sequence $\{Z_{jt}\}$ represents repeated outcomes of a stochastic detection mechanism induced by the combined selection-and-testing procedure across perturbation iterations.

We model this detection mechanism using a Bernoulli working model with effective detection probabilities
\[
Z_{jt} \mid H_0 \sim \text{Bernoulli}(\pi_0), \qquad
Z_{jt} \mid H_1 \sim \text{Bernoulli}(\pi_1),
\]
where
\[
\pi_0 = \Prob(Z_{jt}=1 \mid H_0), \qquad
\pi_1 = \Prob(Z_{jt}=1 \mid H_1),
\]
and $\pi_1 > \pi_0$.

This model is a working approximation: the detection indicators are neither independent across iterations nor generated by a true Bernoulli process. Rather, the model provides a parsimonious representation of the marginal detection behavior under relevance and irrelevance, which enables a tractable likelihood-ratio-type evidence construction in the next subsection.
The conditional independence assumption is adopted as a simplifying device to obtain a tractable likelihood-ratio representation; the resulting evidence measure is therefore an approximation whose practical validity is assessed empirically.

The parameters $\pi_0$ and $\pi_1$ summarize the effective detection behavior of the full two-stage procedure. In particular, $\pi_0$ reflects both the probability that an irrelevant variable enters the candidate set and the probability that it is subsequently confirmed. These quantities are therefore treated as working parameters characterizing the empirical detection behavior of the procedure.

\subsection{Likelihood-Ratio-Based Sequential Evidence Aggregation}
\label{sec:evidence}
For each variable $j$, the sequential aggregation problem is formulated as a comparison between the hypotheses
\[
H_{0j}: \text{variable } j \text{ is irrelevant},
\qquad
H_{1j}: \text{variable } j \text{ is relevant}.
\]
Under the working detection model from Section~\ref{sec:detection}, the sequence of binary detection indicators $\{Z_{js}\}_{s=1}^{t}$ has hypothesis-dependent working distributions. For a single perturbation iteration $s$, define the likelihood ratio
\[
L_{js}
=
\frac{\Prob(Z_{js}\mid H_{1j})}{\Prob(Z_{js}\mid H_{0j})}
=
\frac{\pi_1^{Z_{js}}(1-\pi_1)^{1-Z_{js}}}
     {\pi_0^{Z_{js}}(1-\pi_0)^{1-Z_{js}}}.
\]

Under the Bernoulli working model, $L_{js}$ compares the relevance
hypothesis $H_{1j}$ to the irrelevance hypothesis $H_{0j}$ for the single
observation $Z_{js}$. If the detection indicators were conditionally
independent across perturbation iterations, the joint likelihood ratio
for $\{Z_{js}\}_{s=1}^t$ would factorize as the product of the
per-iteration likelihood ratios. This motivates defining the cumulative
quantity
\[
E_{jt} = \prod_{s=1}^{t} L_{js},
\]
which aggregates evidence sequentially across perturbations. Under equal
prior weights, each $L_{js}$ corresponds to a Bayes factor under the
working model, so $E_{jt}$ admits a Bayes-factor–type interpretation.

For numerical stability, we work with the logarithmic form
\begin{equation}
\label{eq:logevidence}
\log E_{jt}
=
\sum_{s=1}^{t}
\left[
Z_{js}\log\!\left(\frac{\pi_1}{\pi_0}\right)
+
(1-Z_{js})\log\!\left(\frac{1-\pi_1}{1-\pi_0}\right)
\right],
\end{equation}
initialized at $\log E_{j0}=0$. Each perturbation iteration contributes
additively: a positive detection adds
$\log(\pi_1/\pi_0)>0$, while a non-detection adds
$\log((1-\pi_1)/(1-\pi_0))<0$. Hence the evidence process can be
updated recursively as
\[
\log E_{jt} = \log E_{j,t-1} + \log L_{jt}.
\]

This representation shows that the evidence process is a weighted
cumulative sum of detection outcomes, where detections and
non-detections contribute asymmetric increments determined by
$(\pi_0,\pi_1)$.

A useful summary of the working evidence model is the break-even detection frequency, that is, the detection probability at which the expected increment of the log-evidence process is exactly zero. If a variable has effective marginal detection probability $q$, then the expected log-evidence increment under the working model is
\[
m(q;\pi_0,\pi_1)
=
q\log\!\left(\frac{\pi_1}{\pi_0}\right)
+
(1-q)\log\!\left(\frac{1-\pi_1}{1-\pi_0}\right).
\]
This quantity is positive if detections occur sufficiently often and negative otherwise. The critical value at which the drift changes sign is
\[
q^\ast(\pi_0,\pi_1)
=
\frac{
\log\!\left(\frac{1-\pi_0}{1-\pi_1}\right)
}{
\log\!\left(\frac{\pi_1(1-\pi_0)}{\pi_0(1-\pi_1)}\right)
}.
\]
Thus, $q^\ast(\pi_0,\pi_1)$ is the detection frequency for which the expected log-evidence increment is zero. Variables with effective detection rates above $q^\ast$ have positive expected log-evidence drift, whereas variables with detection rates below $q^\ast$ have negative expected drift under the working model. In the implementation, this quantity is used as a diagnostic for whether a given calibration of $(\pi_0,\pi_1)$ yields an overly permissive evidence process.

Because $(\pi_0,\pi_1)$ are calibrated empirically and perturbation iterations are not independent, the resulting evidence measure should be interpreted as a working-model approximation rather than as an exact Bayesian posterior-odds update. The construction nevertheless provides a tractable and interpretable score that mimics likelihood-ratio accumulation.

\begin{remark}[Drift separation]
\label{rem:drift}
Under fixed $(\pi_0,\pi_1)$ with $\pi_1>\pi_0$ and conditional
independence across perturbation iterations, the expected increment of
$\log E_{jt}$ is positive under $H_1$ and negative under $H_0$. More
precisely,
\[
\mathbb{E}[\log L_{js}\mid H_1]
=
D_{KL}\!\left(
\mathrm{Bernoulli}(\pi_1)\,\|\,\mathrm{Bernoulli}(\pi_0)
\right)
>0
\]
and
\[
\mathbb{E}[\log L_{js}\mid H_0]
=
-\,D_{KL}\!\left(
\mathrm{Bernoulli}(\pi_0)\,\|\,\mathrm{Bernoulli}(\pi_1)
\right)
<0.
\]
Thus, under the working model, the cumulative log-evidence drifts
upward for relevant variables and downward for irrelevant variables.
This property links the proposed procedure to classical
likelihood-ratio–based sequential testing methods, such as the
SPRT.
\end{remark}

The speed of evidence accumulation is governed by the separation
between $\pi_1$ and $\pi_0$. Larger separation implies larger expected
increments in absolute value and therefore faster threshold crossing.
Accordingly, the proposed procedure is best viewed as a likelihood-ratio–motivated scoring and stopping rule for repeated detection outcomes, rather than as a formally calibrated Bayesian or frequentist testing procedure.
Relative to simple frequency thresholding, the likelihood-ratio formulation provides asymmetric weighting of detections and non-detections and yields a natural stopping rule based on cumulative evidence rather than a fixed number of iterations.

The evidence process also admits a decision-theoretic interpretation under the working model. In particular, the next proposition shows that thresholding the cumulative evidence is equivalent to the Bayes-optimal classification rule when false inclusion and false exclusion are assigned explicit losses and prior relevance probabilities are specified.

\begin{proposition}[Decision-theoretic interpretation of the evidence rule]
\label{prop:decision}
Consider a variable $j$ and suppose that, under the working model,
\[
Z_{jt}\mid H_h \sim \mathrm{Bernoulli}(\pi_h),
\qquad h\in\{0,1\},
\]
with $\pi_1>\pi_0$, and that $\{Z_{jt}\}_{t=1}^T$ are conditionally independent across $t$ given the hypothesis. Let
\[
E_{jT}
=
\prod_{t=1}^T
\frac{\Prob(Z_{jt}\mid H_1)}{\Prob(Z_{jt}\mid H_0)}
\]
denote the cumulative likelihood ratio.

Let $\ell_1>0$ denote the loss of false exclusion ($H_1 \to H_0$) and $\ell_0>0$ the loss of false inclusion ($H_0 \to H_1$). Then the Bayes-optimal rule is
\[
\text{decide } H_1
\quad\Longleftrightarrow\quad
\log E_{jT}
>
\log\!\left(
\frac{\ell_0\,\Pr(H_0)}{\ell_1\,\Pr(H_1)}
\right).
\]
Under equal prior probabilities, this simplifies to
\[
\text{decide } H_1
\quad\Longleftrightarrow\quad
\log E_{jT}
>
\log\!\left(\frac{\ell_0}{\ell_1}\right).
\]
\end{proposition}

A proof is provided in Appendix~\ref{app:theory}. As with the other theoretical results in this section, the proposition is derived under the idealized working model and should therefore be interpreted as an approximation to the behavior of the empirical procedure rather than as a finite-sample guarantee.

Proposition~\ref{prop:decision} shows that, under the working detection model, the cumulative evidence measure $E_{jT}$ fully determines the Bayes-optimal classification of variable relevance. Thus, the proposed threshold rule is not merely heuristic: it coincides with the optimal decision rule under asymmetric misclassification costs and prior relevance probabilities.

In particular, symmetric loss and prior specifications imply a threshold at zero on the log-evidence scale, while more conservative inclusion rules correspond to higher relative loss for false inclusion or higher prior odds in favor of irrelevance. The result therefore provides a direct mapping between threshold choice and the implied decision-theoretic trade-off between false inclusion and false exclusion.

\subsection{Behavior of the Evidence Process}
\label{sec:theory}

To provide intuition for the proposed evidence process, consider an idealized setting in which the detection indicators $\{Z_{jt}\}_{t\geq 1}$ are independent and identically distributed with
\[
Z_{jt} \sim \mathrm{Bernoulli}(q_j),
\]
where $q_j = \Prob(Z_{jt}=1)$ denotes the marginal detection probability.

Under the working model, the expected increment of the log-evidence process is
\[
m(q_j;\pi_0,\pi_1)
=
q_j\log\!\left(\frac{\pi_1}{\pi_0}\right)
+
(1-q_j)\log\!\left(\frac{1-\pi_1}{1-\pi_0}\right).
\]
This quantity determines the direction of evidence accumulation. In particular, the drift is positive if $q_j > q^\ast(\pi_0,\pi_1)$ and negative otherwise, where $q^\ast(\pi_0,\pi_1)$ is the break-even detection probability defined in Section~\ref{sec:evidence}.

The following proposition makes this intuition precise by showing that the sign of the drift determines the asymptotic classification outcome under the idealized working model.

\begin{proposition}
\label{prop:classification}
Under the idealized working model, if $q_j > q^\ast(\pi_0,\pi_1)$, then $\log E_{jt} \to +\infty$ almost surely, whereas if $q_j < q^\ast(\pi_0,\pi_1)$, then $\log E_{jt} \to -\infty$ almost surely. Consequently, for any fixed threshold, the probability of correct classification converges to one as the number of perturbation iterations increases.
\end{proposition}
A proof is provided in Appendix~\ref{app:theory}.

The result provides intuition for why the proposed score can separate persistent detections from persistent non-detections under the working model. Variables with stronger separation from the break-even detection rate are classified more rapidly, while variables near the boundary require more iterations.

Under the same idealized working model, the expected stopping time is approximately inversely proportional to the absolute drift of the log-evidence process; a formal statement is provided in Appendix~\ref{app:theory}.

The formal proofs and additional results on boundary-crossing probabilities are provided in Appendix~\ref{app:theory}. The practical behavior of the evidence process therefore depends on how the working detection probabilities are calibrated, which is discussed next.

\subsection{Calibration of Detection Probabilities}
\label{sec:calibration}

The evidence process requires specification of the working detection probabilities $(\pi_0,\pi_1)$. We calibrate these parameters in a short pilot phase and hold them fixed during subsequent evidence accumulation.

Let $T_{\mathrm{pilot}}$ denote the number of pilot iterations and define the pilot detection frequencies
\[
\hat{\pi}_{j,\mathrm{pilot}}
=
\frac{1}{T_{\mathrm{pilot}}}
\sum_{t=1}^{T_{\mathrm{pilot}}} Z_{jt}.
\]

We estimate $\pi_0$ from variables with low pilot detection frequencies and $\pi_1$ from variables with high pilot detection frequencies, using lower- and upper-quantile subsets of $\{\hat{\pi}_{j,\mathrm{pilot}}\}_{j=1}^p$.

To ensure stable behavior of the evidence process, we regularize $\hat{\pi}_0$ by shrinking it toward the nominal level and imposing a lower bound. This prevents excessively aggressive evidence accumulation when the raw pilot estimate is very small.

The resulting calibrated values $(\hat{\pi}_0,\hat{\pi}_1)$ define the likelihood-ratio increments used in the sequential evidence process. Additional calibration strategies and sensitivity analyses are discussed in Appendix~\ref{app:calibration}.

\subsection{Stopping Rule and Variable Classification}
\label{sec:stopping}

The sequential evidence process provides a variable-specific rule for
classifying controls as relevant or irrelevant as perturbation
iterations proceed. Let $E_{jt}$ denote the cumulative evidence in
favor of $H_1$ for variable $j$ after $t$ evidence iterations. Because
$E_{jt}$ is constructed as a likelihood-ratio--based evidence measure
under the working detection model, variable classification can be
formulated through evidence thresholds.

Specifically, for thresholds $0<\tau_0<1<\tau_1$, a variable is
classified as relevant once
\[
E_{jt}\geq \tau_1,
\]
and as irrelevant once
\[
E_{jt}\leq \tau_0.
\]
Thus, $\tau_1$ is the upper evidence threshold required for classifying
a variable as relevant, whereas $\tau_0$ is the lower evidence threshold
for classifying a variable as irrelevant.

Equivalently, in logarithmic form,
\[
\log E_{jt}\geq c_1
\qquad\text{or}\qquad
\log E_{jt}\leq c_0,
\]
where $c_1=\log\tau_1$ and $c_0=\log\tau_0$.

In the baseline implementation, we use symmetric log-thresholds
$c_1=c$ and $c_0=-c$ for some $c>0$, so that
\[
\tau_1=e^c
\qquad\text{and}\qquad
\tau_0=e^{-c}.
\]
Thus, $c$ is the required magnitude of cumulative log-evidence for classification: a variable is classified as relevant once its log-evidence reaches $+c$ and as irrelevant once its log-evidence reaches $-c$. Equivalently, on the evidence scale, $e^c$ is the upper evidence threshold in favor of relevance and $e^{-c}$ is the lower threshold in favor of irrelevance.

Hence, larger values of $c$ require stronger accumulated evidence before a classification is made, leading to more conservative decisions and typically longer perturbation runs, whereas smaller values of $c$ allow earlier classifications but may be less stable. In the simulation study, we consider values such as $c=\log(3)$, $c=\log(10)$, and $c=\log(30)$, corresponding to increasingly stringent evidence requirements.

\subsection{Practical Remarks}
\label{sec:limitations}

The proposed procedure relies on a working model in which detection indicators are treated as approximately independent across perturbation iterations. In practice, bootstrap resampling and stochastic imputation induce dependence, so the resulting evidence measure should be interpreted as an approximate likelihood-ratio score rather than as an exact Bayesian quantity. The calibration of $(\pi_0,\pi_1)$ is empirical and may affect finite-sample behavior; its performance is therefore assessed through simulation and sensitivity analysis. Additional implementation details are provided in Appendix~\ref{app:calibration}.

The preceding discussion developed the proposed procedure in modular form by introducing the detection rule, the working evidence model, the calibration strategy, and the stopping rule in turn. Algorithm~1 brings these elements together and summarizes the full practical implementation of the BOOT-MI sequential evidence procedure.

\begin{tcolorbox}[title={Algorithm 1: BOOT-MI Sequential Evidence Procedure},
                  colback=white,
                  colframe=black,
                  breakable]

\textbf{Input:} Incomplete data with outcome $Y$, variable of interest $D$, and candidate controls $X$; pilot length $T_{\mathrm{pilot}}$; maximum iterations $T_{\max}$; threshold $c>0$; minimum classification delay $t_{\min}$; significance level $\alpha$; calibration constants $\lambda_0$, $\pi_{0,\min}$, and $q^\ast_{\min}$.

\vspace{0.3em}

\textbf{Step 1: Pilot calibration.}

For $t=1,\dots,T_{\mathrm{pilot}}$:
\begin{enumerate}[itemsep=0.01em, topsep=0.1em]
\item Draw a bootstrap sample from the incomplete dataset.
\item Perform stochastic imputation to obtain a completed dataset.
\item Construct the PDS candidate set $S_t = S_{Y,t} \cup S_{D,t}$, where $S_{Y,t}$ and $S_{D,t}$ denote the variables selected by LASSO in the outcome equation and the auxiliary equation for $D$, respectively.
\item For each $j=1,\dots,p$, construct the asymmetric detection indicators:
\begin{align*}
Z_{jt} &= \mathbbm{1}\left\{ j \in S_{D,t} \;\lor\; \right. \\ &\hspace{1cm} \left.\bigl( j \in S_{Y,t} \;\land\; \text{coefficient on } X_j \text{ significant in } Y \sim D + X_{S_t} \text{ at level } \alpha \bigr) \right\}
\end{align*}
\end{enumerate}

Estimate $\hat{\pi}_0^{\mathrm{raw}}$ and $\hat{\pi}_1$ from the pilot
detection frequencies. Form the stabilized null calibration
\[
\tilde{\pi}_0
=
\max\!\left\{
\pi_{0,\min},
(1-\lambda_0)\hat{\pi}_0^{\mathrm{raw}} + \lambda_0 \alpha
\right\},
\]
and, if necessary, increase $\tilde{\pi}_0$ minimally until
$q^\ast(\hat{\pi}_0,\hat{\pi}_1)\ge q^\ast_{\min}$. Use the resulting
$\hat{\pi}_0$ and $\hat{\pi}_1$ in the evidence step.

\vspace{0.3em}

\textbf{Step 2: Initialize evidence.}

Set $\log E_{j0}=0$ for all $j=1,\dots,p$ and mark all variables as undecided.

\vspace{0.3em}

\textbf{Step 3: Sequential evidence accumulation.}

For $t=1,\dots,T_{\max}$:
\begin{enumerate}[itemsep=0.01em, topsep=0.1em]
\item Draw a bootstrap sample from the incomplete dataset.
\item Perform stochastic imputation to obtain a completed dataset.
\item Apply the PDS procedure.
\item Construct detection indicators $Z_{jt}$.
\item Update the cumulative log-evidence for each variable $j$:
\[
\log E_{jt}
=
\log E_{j,t-1}
+
Z_{jt}\log\!\left(\frac{\hat{\pi}_1}{\hat{\pi}_0}\right)
+
(1-Z_{jt})\log\!\left(\frac{1-\hat{\pi}_1}{1-\hat{\pi}_0}\right).
\]
\item If $t \ge t_{\min}$, classify variable $j$:
\[
\log E_{jt} \geq c \;\Rightarrow\; \text{relevant}, \qquad
\log E_{jt} \leq -c \;\Rightarrow\; \text{irrelevant}.
\]
\end{enumerate}

\textbf{Step 4: Stopping rule.}

Terminate when all variables are classified or when $t=T_{\max}$.

\vspace{0.3em}

\textbf{Output:} Set of variables classified as relevant.

\end{tcolorbox}

\section{Simulation Study}
\label{sec:simulation}

This section evaluates the proposed sequential evidence aggregation procedure in a Monte Carlo simulation study. The simulation is designed to address four questions: whether the method improves variable-selection accuracy relative to standard aggregation rules, whether the asymmetric detection rule preserves estimation performance, how sensitive the procedure is to the calibration of $(\pi_0,\pi_1)$ and the threshold $c$, and whether sequential stopping yields computational gains relative to fixed-budget alternatives.

The simulation builds on the data-generating processes of \citet{BelloniCH2014a}, extended to incorporate missing data and perturbation-based estimation. We generate a synthetic population of size $N=10{,}000$ and draw samples of size $n \in \{100,500,1000\}$. Each scenario is evaluated using 500 Monte Carlo replications.

\subsection{Design}
\label{sec:sim_design}

The main simulation design follows the partially linear model
\begin{align}
Y_i &= D_i \alpha_0 + X_i' \beta + \varepsilon_i, \\
D_i &= X_i' \gamma + \nu_i,
\end{align}
where $(\varepsilon_i,\nu_i)$ are independent standard normal random variables and $X_i \sim \mathcal{N}(0,\Sigma)$ with $\Sigma_{jl}=0.5^{|j-l|}$. The covariate vector has dimension $p=50$, of which $k_0=5$ variables are relevant. The coefficient vectors satisfy
\[
\beta_j = \gamma_j =
\begin{cases}
\kappa j^{-2}, & j \leq k_0, \\
0, & \text{otherwise},
\end{cases}
\]
where $\kappa$ is chosen to achieve target coefficients of determination $R^2 \in \{0.2,0.6\}$. We consider both homoscedastic and heteroscedastic error specifications. Missingness is introduced under MCAR, MAR, and MNAR mechanisms at rates of $20\%$, $40\%$, and $60\%$. As a robustness check, we additionally consider an extended design with group-specific nuisance components in the outcome equation.

Table~\ref{tab:sim_overview} summarizes the resulting simulation scenarios. The design spans both favorable and difficult settings for perturbation-based variable selection by varying sample size, signal strength, error structure, missing-data mechanism, and missingness severity. This allows us to assess not only average selection performance, but also the robustness of the proposed calibration and stopping rule under increasingly challenging conditions. Because Design~2 is included as a robustness check rather than as part of the full factorial design, the simulation evidence is driven primarily by the 108 scenarios from Design~1.

\begin{table}
 \caption{Overview on the Simulation Scenarios}
 \label{tab:sim_overview}
 \centering
 \begin{minipage}{\linewidth}
 \footnotesize
  The table summarizes the simulation scenarios used in the Monte Carlo study. For expositional convenience, we refer to the clean Belloni-type design without grouping variables as Design~1 and to the extended design with grouping nuisance terms as Design~2. In the implementation, these correspond to separate data-generating-process specifications, with the full factorial grid applied to Design~1 and a reduced robustness grid ($n=500$, homoscedastic errors) applied to Design~2.\\
  \end{minipage}

 \begin{tabular}{llc}
\toprule
\textbf{Decision} & \textbf{Scenario} & \# Scenarios\\
\midrule
\multirow{2}{*}{DGP Design} & Design 1: without grouping variables &\multirow{2}{*}{2}\\
& Design 2: with grouping nuisance terms \\
\midrule
\multirow{3}{*}{Sample Size} & $n=100$ &\multirow{3}{*}{3}\\
& $n=500$ \\
& $n=1000$ \\
\midrule
\multirow{2}{*}{Variance of DGP} & homoscedastic &\multirow{2}{*}{2}\\
& heteroscedastic \\
\midrule
\multirow{2}{*}{$R^2$ of DGP} & $R^2 = 20\%$ &\multirow{2}{*}{2}\\
& $R^2 = 60\%$\\
\midrule
\multirow{3}{*}{Missing Responses} & missing completely at random (MCAR) &\multirow{3}{*}{3}\\
& missing at random (MAR)\\
& missing not at random (MNAR)\\
\midrule
\multirow{3}{*}{Degree of Missing Data} & 20\% missing observations &\multirow{3}{*}{3}\\
& 40\% missing observations \\
& 60\% missing observations \\
\midrule
\multicolumn{2}{l}{Main scenarios (Design 1: $3 \times 2 \times 2 \times 3 \times 3$)} & 108\\
\multicolumn{2}{l}{Robustness scenarios (Design 2: $1 \times 1 \times 2 \times 3 \times 3$)} & 18\\
\midrule
Total&&126 scenarios\\
\bottomrule
\end{tabular}
\end{table}

\subsection{Implementation and Evaluation}
\label{sec:sim_metrics}

We implement the proposed method with pilot calibration ($T_{\mathrm{pilot}}=20$), maximum evidence iterations $T_{\max}=200$, and minimum classification delay $t_{\min}=5$. Unless stated otherwise, the simulation results use the baseline threshold $c=\log(10)$. Sensitivity analyses additionally consider $c\in\{\log(3),\log(10),\log(30)\}$. The baseline implementation uses the stabilized pilot-based calibration described in Section~\ref{sec:calibration}; nominal and permutation-based calibrations are considered only as robustness checks.

We compare the proposed method to three benchmark aggregation rules: the union rule and frequency thresholding at 50\% and 75\%. Benchmarks are evaluated both at a fixed budget ($T=200$) and at matched budgets corresponding to the stopping time of the proposed method.

We evaluate performance in terms of variable selection accuracy (TPR, FPR, precision, model size), estimation performance (bias, RMSE, coverage), and computational efficiency (number of perturbation iterations until stopping).

\subsection{Results}
\label{sec:results}

We organize the results around variable selection, treatment-effect estimation, and computational efficiency. Throughout, we compare the proposed sequential evidence method to the union rule and to frequency-threshold aggregation at the 50\% and 75\% levels. For the frequency-based benchmarks, we report both fixed-budget and matched-budget comparisons.

\subsubsection{Variable Selection}

Table~\ref{tab:selection} summarizes variable selection performance under the fixed-budget comparison. The union rule achieves perfect sensitivity by construction but selects all variables. The 75\% threshold is highly selective but suffers from low sensitivity. The proposed method achieves the highest TPR among the selective methods, at the cost of a higher FPR. The 50\% threshold provides a more conservative alternative with lower FPR but reduced sensitivity.

\begin{table}[t]
\caption{Variable Selection Performance: Fixed-Budget Comparison}
\label{tab:selection}
\centering
\begin{tabular}{lcccc}
\toprule
Method & TPR & FPR & Dist. to Ideal & Model Size \\
\midrule
Union Rule & 1.000 & 1.000 & 1.000 & 50.0 \\
Frequency threshold (50\%) & 0.636 & 0.076 & 0.385 & 6.6 \\
Frequency threshold (75\%) & 0.459 & 0.006 & 0.542 & 2.6 \\
Proposed Method & 0.774 & 0.384 & 0.486 & 21.1 \\
\bottomrule
\end{tabular}
\end{table}

Results are nearly unchanged under the matched-budget comparison in Table~\ref{tab:selection_matched}, indicating that the gains are not driven solely by longer runs of the benchmark methods.

\begin{table}[t]
\caption{Variable Selection Performance: Matched-Budget Comparison}
\label{tab:selection_matched}
\centering
\begin{tabular}{lcccc}
\toprule
Method & TPR & FPR & Dist. to Ideal & Model Size \\
\midrule
Union Rule (matched) & 1.000 & 1.000 & 1.000 & 50.0 \\
Frequency threshold (50\%, matched) & 0.635 & 0.083 & 0.388 & 6.9 \\
Frequency threshold (75\%, matched) & 0.461 & 0.008 & 0.540 & 2.6 \\
Proposed Method & 0.774 & 0.384 & 0.486 & 21.1 \\
\bottomrule
\end{tabular}
\end{table}

Figure~\ref{fig:frontier_fixed_overall} shows the aggregate TPR--FPR trade-off under the fixed-budget comparison, while Figure~\ref{fig:frontier_matched_overall} reports the matched-budget analogue. In both cases, the proposed method occupies a distinctly higher-TPR position than the 50\% threshold, while remaining far less inclusive than the union rule.

\begin{figure}[t]
\centering
\caption{TPR--FPR Frontier: Fixed-Budget Comparison (Overall)}
\begin{minipage}{0.9\linewidth}
\footnotesize
Each point represents the average TPR and FPR of a method across all simulation scenarios under the fixed-budget comparison ($T=200$ iterations for all benchmark methods). The cross marks the ideal point at $(0,1)$.
\end{minipage}
\label{fig:frontier_fixed_overall}
\includegraphics[width=0.85\linewidth]{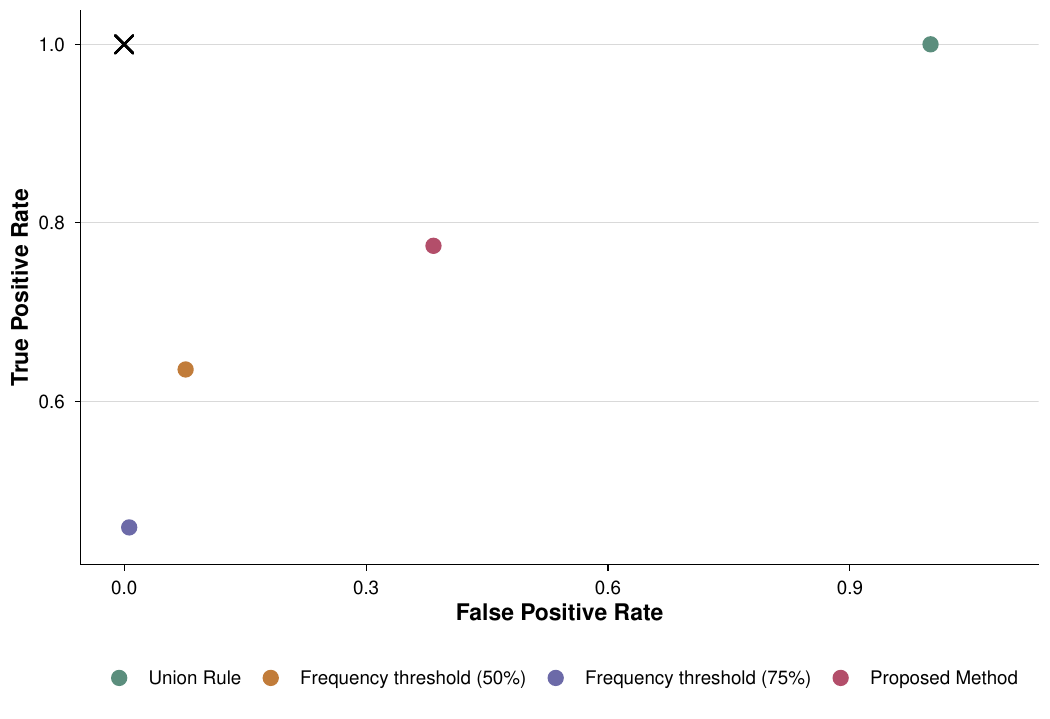}
\end{figure}

\begin{figure}[t]
\centering
\caption{TPR--FPR Frontier: Matched-Budget Comparison (Overall)}
\begin{minipage}{0.9\linewidth}
\footnotesize
Each point represents the average TPR and FPR of a method across all simulation scenarios under the matched-budget comparison, in which benchmark methods are restricted to the same average number of iterations as the proposed sequential method.
\end{minipage}
\label{fig:frontier_matched_overall}
\includegraphics[width=0.85\linewidth]{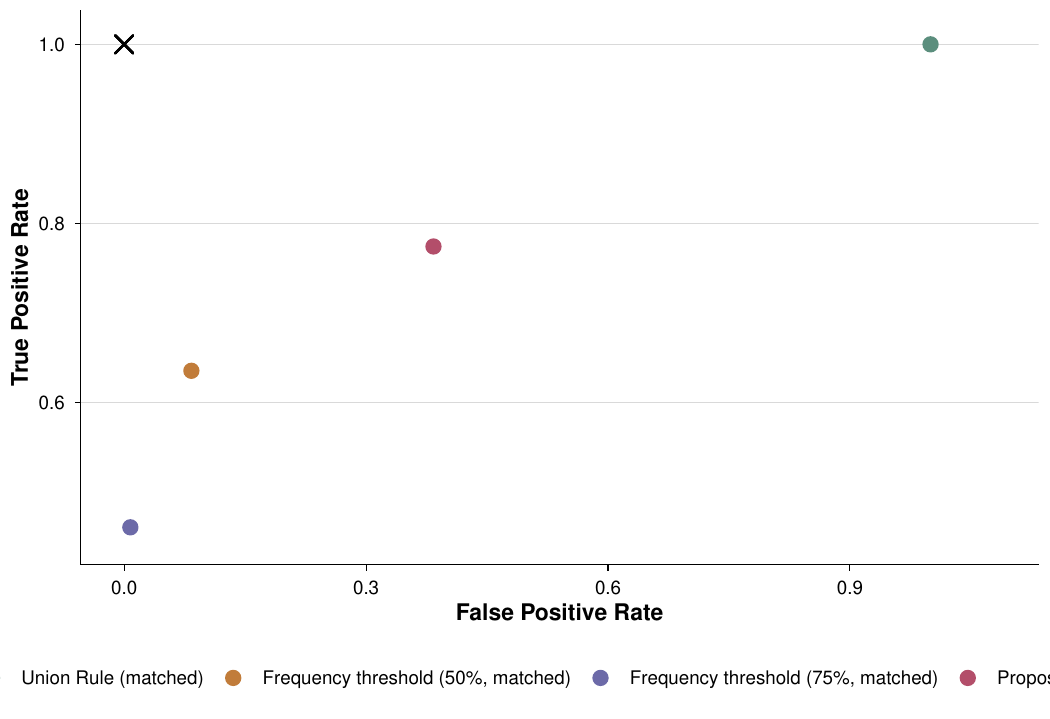}
\end{figure}

Figure~\ref{fig:tpr_fpr_comparison_fixed} shows that the proposed method consistently attains a higher TPR than the 50\% threshold across all scenario dimensions. The gain is largest in more difficult settings, especially small samples and MNAR missingness, while the higher FPR reflects the method's more permissive treatment of persistent but weaker signals.

\begin{figure}[t]
\centering
\caption{TPR and FPR by Scenario Dimension: Proposed Method vs.\ 50\% Frequency Threshold (Fixed-Budget)}
\begin{minipage}{0.9\linewidth}
\footnotesize
Mean TPR (top) and FPR (bottom) for the proposed method and the $50\%$ frequency threshold, broken down by sample size, missingness rate, missing-data mechanism, signal strength, and error structure. Error bars are omitted because Monte Carlo confidence intervals are too small to be visually distinguishable at the scale of the figure.
\end{minipage}
\label{fig:tpr_fpr_comparison_fixed}
\includegraphics[width=\linewidth]{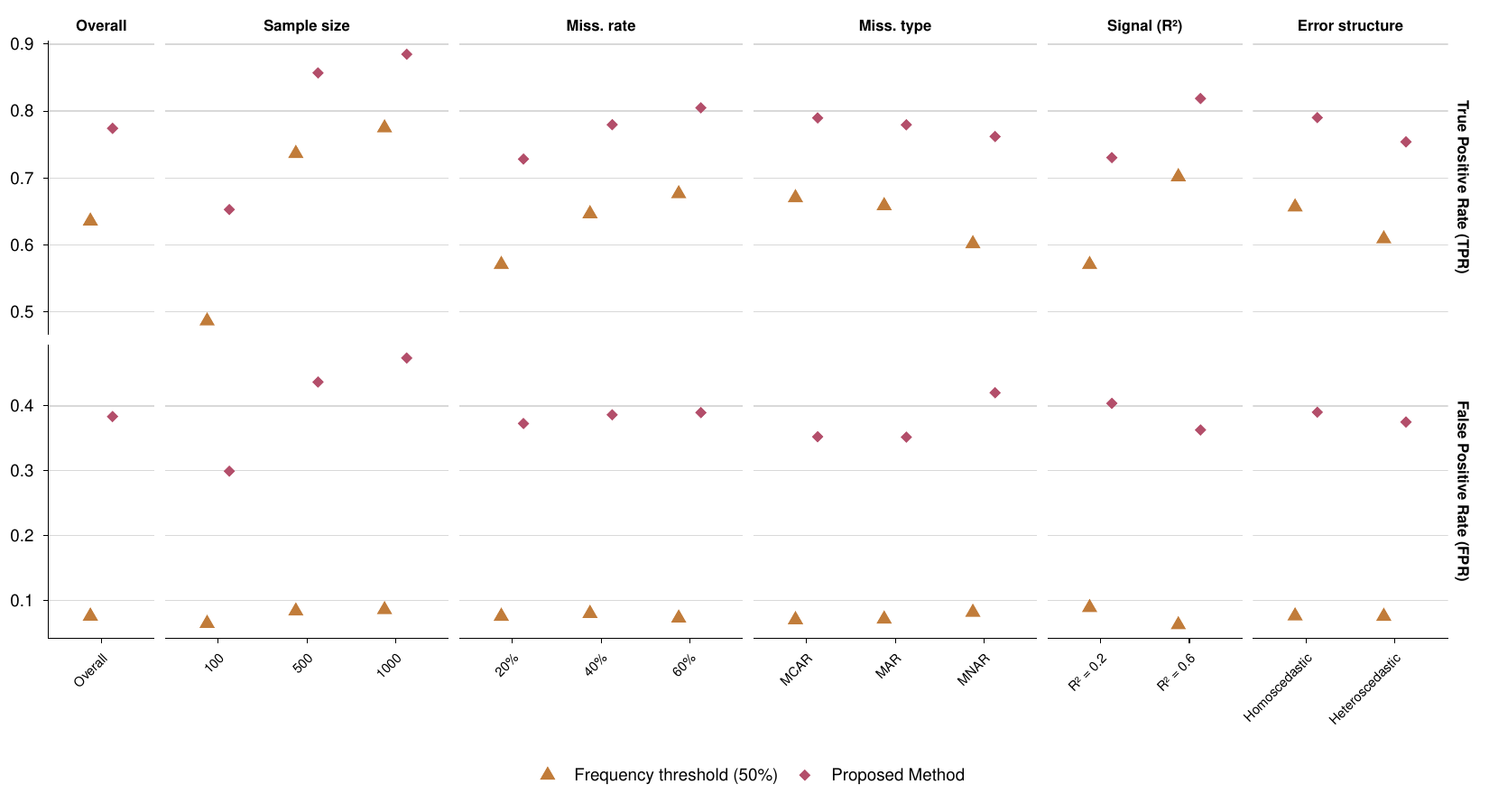}
\end{figure}

\subsubsection{Treatment Effect Estimation}
\label{sec:treatment_sim}

Table~\ref{tab:treatment} reports treatment-effect performance under the fixed-budget comparison. The proposed method achieves bias and RMSE comparable to the 50\% threshold while providing the highest coverage among the selective methods. By contrast, the 75\% threshold exhibits the largest bias and lowest coverage, reflecting the inferential cost of omitting relevant controls. Results under the matched-budget comparison are closely similar (Table~\ref{tab:treatment_matched}).

\begin{table}[t]
\caption{Treatment Effect Estimation: Fixed-Budget Comparison}
\label{tab:treatment}
\centering
\begin{tabular}{lcccc}
\toprule
Method & Bias & RMSE & Coverage & Model Size \\
\midrule
Union Rule & 0.012 & 0.191 & 0.867 & 50.0 \\
Frequency threshold (50\%) & 0.029 & 0.183 & 0.845 & 6.6 \\
Frequency threshold (75\%) & 0.053 & 0.204 & 0.824 & 2.6 \\
Proposed Method & 0.023 & 0.184 & 0.851 & 21.1 \\
\bottomrule
\end{tabular}
\end{table}

\begin{table}[t]
\caption{Treatment Effect Estimation: Matched-Budget Comparison}
\label{tab:treatment_matched}
\centering
\begin{tabular}{lcccc}
\toprule
Method & Bias & RMSE & Coverage & Model Size \\
\midrule
Union Rule (matched) & 0.012 & 0.191 & 0.867 & 50.0 \\
Frequency threshold (50\%, matched) & 0.030 & 0.183 & 0.845 & 6.9 \\
Frequency threshold (75\%, matched) & 0.053 & 0.203 & 0.825 & 2.6 \\
Proposed Method & 0.023 & 0.184 & 0.851 & 21.1 \\
\bottomrule
\end{tabular}
\end{table}

\subsubsection{Computational Efficiency}

Table~\ref{tab:runtime} shows that the proposed stopping rule requires on average 50.7 iterations, compared with the fixed budget of $T=200$ iterations used by the benchmark methods. This corresponds to a reduction of approximately 75\% in the number of perturbation iterations.

\begin{table}[t]
\caption{Computational Efficiency}
\label{tab:runtime}
\centering
\begin{minipage}{\linewidth}
\footnotesize
The runtime column reports the average total runtime per Monte Carlo replication under the current implementation. Because the perturbation history is precomputed before evaluating all aggregation rules, these runtimes are not method-specific and should therefore be interpreted only as a computational reference.\\
\end{minipage}
\begin{tabular}{lccc}
\toprule
Method & Avg Iterations & Std. Iterations & Runtime (s) \\
\midrule
Fixed BOOT-MI (T=200) & 200.0 & 0.0 & 624.3 \\
Proposed Sequential Method & 50.7 & 40.7 & 624.3 \\
\bottomrule
\end{tabular}
\end{table}

Figure~\ref{fig:evidence_paths} illustrates the sequential evidence accumulation for a representative realization. Relevant variables exhibit upward drift and irrelevant variables downward drift, with classification occurring once the paths cross the decision boundaries.

\begin{figure}[t]
\centering
\caption{Example Evidence Paths from a Single Realization}
\begin{minipage}{0.9\linewidth}
\footnotesize
The figure displays cumulative log-evidence paths $\log E_{jt}$ for all $p=50$ candidate variables from a single realization of the BOOT-MI sequential evidence procedure ($n=300$, $R^2=0.6$, $20\%$ MAR, pilot-calibrated $\hat{\pi}_0$). Blue paths correspond to truly relevant variables ($X_1,\dots,X_5$); red paths correspond to irrelevant variables. Dashed horizontal lines indicate the symmetric classification thresholds at $\pm c$. The upper shaded region marks the relevance classification zone ($\log E_{jt}\geq c$) and the lower shaded region the irrelevance zone ($\log E_{jt}\leq -c$). The procedure terminates once all paths have exited the continuation region.
\end{minipage}
\label{fig:evidence_paths}
\includegraphics[width=\linewidth]{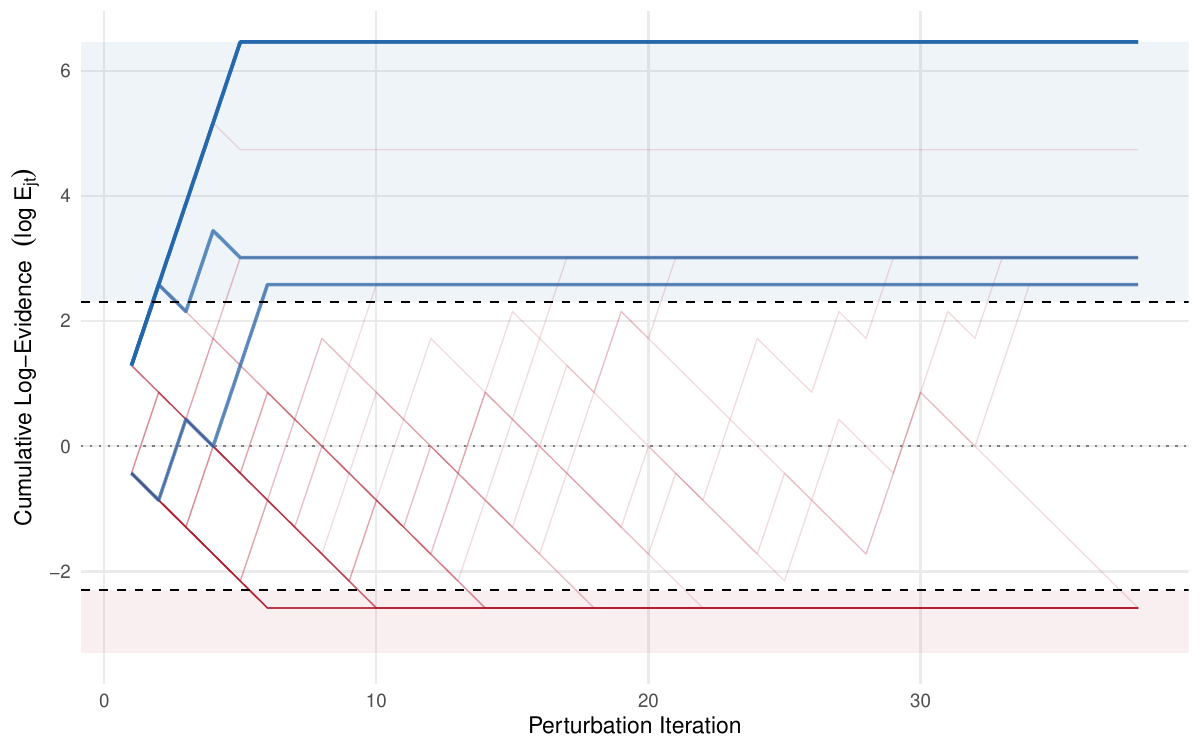}
\end{figure}

Across all simulation scenarios, the pilot calibration is stable and the fallback rule is not triggered.

\subsubsection{Summary of Simulation Findings}

The simulation study yields three main findings. First, the proposed method achieves the highest sensitivity among the selective methods and translates this into favorable estimation performance, especially relative to stricter frequency thresholding. Second, the TPR advantage is robust across simulation settings and is most pronounced in difficult scenarios such as small samples, MNAR missingness, and weak signals. Third, the sequential stopping rule substantially reduces computational cost relative to fixed-budget procedures.

Additional disaggregated results by sample size, missingness rate, missing-data mechanism, signal strength, and error structure are reported in Appendix~\ref{app:sim_tables}.

\section{Empirical Illustration}
\label{sec:empirical}

We illustrate how the proposed BOOT-MI sequential evidence procedure behaves in a realistic applied setting using data from Round 10 of the European Social Survey (ESS), integrated file, edition 3.3. The empirical analysis is based on an analytic subsample of women constructed from the ESS Round 10 integrated dataset. We restrict the sample to observations with nonmissing values for the employment indicator and the household child indicator; the remaining covariates are allowed to contain missing values and are handled through multiple imputation. After applying the sample restrictions described in Appendix~\ref{app:empirical_details}, the resulting analysis sample contains $n=8{,}543$ observations.

Each perturbation iteration consists of bootstrap resampling of the incomplete dataset, stochastic imputation using random forests, and variable selection using the PDS framework. Within each iteration, LASSO is applied to the outcome equation and the auxiliary equation for the variable of interest to construct the PDS candidate set. The asymmetric detection rule from Section~\ref{sec:detection} is then applied: variables selected in the auxiliary equation are retained automatically, while variables selected only in the outcome equation must additionally be significant at level $\alpha=0.05$ in the post-selection outcome regression.

The working detection probabilities $(\pi_0,\pi_1)$ are calibrated from a pilot phase with $T_{\mathrm{pilot}}=20$ iterations. In the ESS application, this yields $\hat{\pi}_{0,\mathrm{raw}}=0.075$, a stabilized value $\hat{\pi}_0=0.0687$, and $\hat{\pi}_1=0.9900$, implying a break-even detection frequency of $q^\ast=0.6296$. In the main empirical specification, we use the more conservative threshold $c=\log(1000)$, check threshold crossing only after $t_{\min}=5$ evidence iterations, and impose a maximum of $T_{\max}=200$ evidence iterations. We use a stricter threshold in the empirical application than in the simulation baseline because the goal there is not method comparison but a more conservative, high-evidence variable classification in a single realized dataset. Below, we report sensitivity to less conservative threshold choices.

As benchmarks, we consider three fixed-iteration aggregation rules based on the raw within-iteration PDS union history. Let $S_t^{\mathrm{PDS}} = S_{Y,t} \cup S_{D,t}$ denote the PDS candidate set in perturbation iteration $t$. The Union benchmark includes a variable if it appears in $S_t^{\mathrm{PDS}}$ in at least one iteration, while the Freq(50\%) and Freq(75\%) benchmarks include a variable if its empirical inclusion frequency across 200 perturbation iterations is at least $0.50$ or $0.75$, respectively.

\paragraph{Variable Selection}
Table~\ref{tab:selection_empirical} reports the variables selected under each aggregation rule. The benchmark methods aggregate repeated raw PDS candidate sets, whereas the proposed method aggregates asymmetric detection indicators through the calibrated sequential evidence process.

The union rule selects all 55 candidate variables and therefore produces the densest specification. The 50\% and 75\% frequency rules yield more parsimonious models with 41 and 31 variables, respectively. The proposed method selects 34 variables and thus lies between the two frequency-based rules in terms of model size.

Substantively, the proposed method retains a stable core set of controls that is also selected by the benchmark rules, including age, years of education, health, household income, subjective financial situation, unemployment experience, several education and marital-status indicators, and a subset of country indicators. Differences across methods arise primarily for weaker attitudinal variables, additional missing-value indicators, urbanicity indicators, and the breadth of included country effects. Overall, the proposed procedure delivers a data-driven compromise between the inclusiveness of the union rule and the stricter frequency thresholds.

\begin{landscape}
 \begin{table}[t]
\centering
\caption{Selected Variables by Aggregation Method}
\label{tab:selection_empirical}
\begin{minipage}{0.95\linewidth}
\small
The table summarizes the variables selected under different aggregation rules applied to the BOOT-MI procedure. A checkmark indicates that at least one variable within the corresponding group is selected by the given method. For categorical variables (e.g., education, marital status, country), selection refers to the inclusion of one or more associated dummy variables. The codes reported in parentheses denote the underlying ESS variables or constructed indicators.\\
\end{minipage}
\renewcommand{\arraystretch}{1.3}
\begin{tabular}{p{6cm} p{3.5cm} p{3.5cm} p{3.5cm} p{3.5cm}}

\toprule
Variable & Proposed & Union & Freq.\ 50\% & Freq.\ 75\% \\
\midrule

Age
& \checkmark {\tiny (agea)}
& \checkmark {\tiny (agea)}
& \checkmark {\tiny (agea)}
& \checkmark {\tiny (agea)} \\

Years of education
& \checkmark {\tiny (eduyrs)}
& \checkmark {\tiny (eduyrs)}
& \checkmark {\tiny (eduyrs)}
& \checkmark {\tiny (eduyrs)} \\

Health
& \checkmark {\tiny (health)}
& \checkmark {\tiny (health)}
& \checkmark {\tiny (health)}
& \checkmark {\tiny (health)} \\

Household income decile
& \checkmark {\tiny (hinctnta)}
& \checkmark {\tiny (hinctnta)}
& \checkmark {\tiny (hinctnta)}
& \checkmark {\tiny (hinctnta)} \\

Subjective financial situation
& \checkmark {\tiny (hincfel)}
& \checkmark {\tiny (hincfel)}
& \checkmark {\tiny (hincfel)}
& \checkmark {\tiny (hincfel)} \\

Happiness / life satisfaction
& \checkmark {\tiny (happy)}
& \checkmark {\tiny (happy, stflife)}
& \checkmark {\tiny (happy)}
& \checkmark {\tiny (happy)} \\

Trust variables
& \checkmark {\tiny (trstlgl)}
& \checkmark {\tiny (ppltrst, trstprl, trstlgl, trstplc)}
& \checkmark {\tiny (ppltrst)}
& \checkmark {\tiny (ppltrst)} \\

Political interest
& \checkmark {\tiny (polintr)}
& \checkmark {\tiny (polintr)}
& \checkmark {\tiny (polintr)}
& \checkmark {\tiny (polintr)} \\

Unemployment experience
& \checkmark {\tiny (uempla)}
& \checkmark {\tiny (uempla)}
& \checkmark {\tiny (uempla)}
& \checkmark {\tiny (uempla)} \\

Education indicators
& \checkmark {\tiny (eisced\_2, 4, 6, 7, NA)}
& \checkmark {\tiny (eisced\_2--7, 55, NA)}
& \checkmark {\tiny (eisced\_2, 4--7, NA)}
& \checkmark {\tiny (eisced\_2, 4, 6, 7)} \\

Marital-status indicators
& \checkmark {\tiny (maritalb\_2--4, 6)}
& \checkmark {\tiny (maritalb\_2--6, NA)}
& \checkmark {\tiny (maritalb\_2--6)}
& \checkmark {\tiny (maritalb\_2--4, 6)} \\

Urbanicity indicators
& \checkmark {\tiny (domicil\_3)}
& \checkmark {\tiny (domicil\_2--5, NA)}
& \checkmark {\tiny (domicil\_3, 5, NA)}
& \checkmark {\tiny (domicil\_3, NA)} \\

Religious belonging
& \checkmark {\tiny (rlgblg\_2)}
& \checkmark {\tiny (rlgblg\_2, NA)}
& \checkmark {\tiny (rlgblg\_2)}
& \checkmark {\tiny (rlgblg\_2)} \\

Country indicators
& \checkmark {\tiny (BG, CH, CZ, EE, FR, GR, IE, IS, IT, ME, NO, PT, SI, SK)}
& \checkmark {\tiny (all)}
& \checkmark {\tiny (BG, CH, CZ, EE, FR, GR, IE, IS, IT, LT, ME, MK, NL, NO, PT, SI, SK)}
& \checkmark {\tiny (BG, CH, EE, FR, GR, IE, IS, IT, NO, PT, SI)} \\

\midrule
Model size
& 34 & 55 & 41 & 31 \\
\bottomrule
\end{tabular}
\renewcommand{\arraystretch}{1.0}
\end{table}
\end{landscape}

Figure~\ref{fig:evidence_empirical} displays the corresponding log-evidence paths. Variables with sustained support across perturbations cross the upper threshold quickly, while variables with weak or unstable support remain near the continuation region or drift downward. In the baseline ESS specification, all 55 candidate variables are classified within 10 perturbation iterations.

\begin{figure}[t]
\centering
\caption{Empirical Evidence Paths}
\label{fig:evidence_empirical}
\begin{minipage}{0.9\linewidth}
\footnotesize
The figure displays cumulative log-evidence paths $\log E_{jt}$ for all $p=55$ candidate variables from the BOOT-MI sequential evidence procedure applied to the ESS data ($n=8{,}543$, pilot-calibrated $\hat{\pi}_0$, $c=\log(1000)$). Paths corresponding to variables classified as relevant are shown with higher opacity; paths of irrelevant variables are faded. Dashed horizontal lines indicate the symmetric classification thresholds at $\pm c$. The procedure classifies all variables within 10 iterations.\\
\end{minipage}
\includegraphics[width=0.9\linewidth]{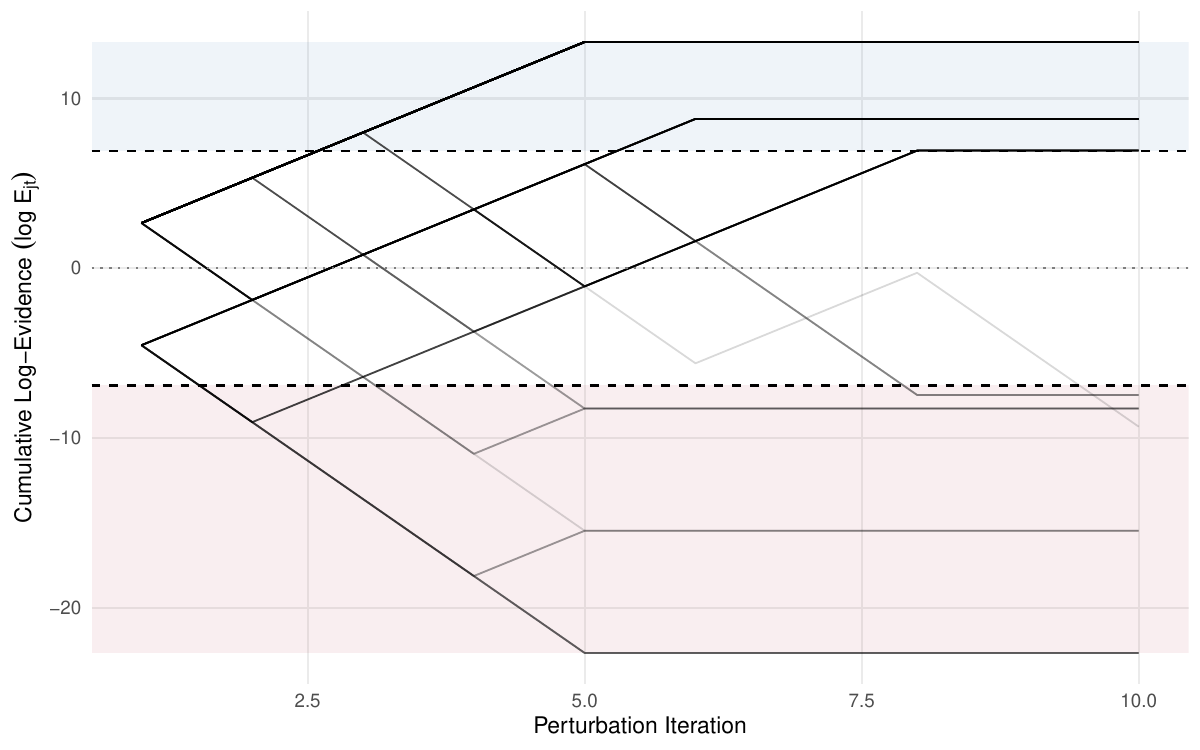}
\end{figure}

\paragraph{Estimated Coefficient on the Variable of Interest}

Table~\ref{tab:treatment_empirical} reports the estimated coefficient on the variable of interest across aggregation methods. For each rule, the selected control set is fixed first, and the coefficient is then re-estimated across $M=10$ multiply imputed datasets using Rubin's rules.

\begin{table}[t]
\centering
\caption{Treatment-Effect Estimates}
\label{tab:treatment_empirical}
\begin{tabular}{lcccc}
\toprule
Method & Estimate & Std.\ Error & 95\% CI & Model Size \\
\midrule
Proposed     & $-0.001$ & $0.011$ & $[-0.023,\; 0.022]$ & 34 \\
Union        & $\phantom{-}0.002$ & $0.011$ & $[-0.021,\; 0.024]$ & 55 \\
Freq (50\%)  & $\phantom{-}0.001$ & $0.011$ & $[-0.022,\; 0.023]$ & 41 \\
Freq (75\%)  & $-0.001$ & $0.012$ & $[-0.023,\; 0.022]$ & 31 \\
\bottomrule
\end{tabular}
\end{table}

The estimated coefficient is close to zero and statistically insignificant under all four aggregation rules. Point estimates range from $-0.001$ to $0.002$, standard errors range from $0.011$ to $0.012$, and the confidence intervals overlap almost completely. Thus, in this application, the estimated conditional association is highly robust to the choice of aggregation rule despite substantial differences in model size. A more detailed substantive discussion is provided in Appendix~\ref{app:empirical_details}.

\paragraph{Sensitivity Analysis}

Table~\ref{tab:sensitivity_empirical} summarizes the sensitivity of the empirical results to alternative evidence thresholds and calibration choices.

\begin{table}[t]
\centering
\caption{Sensitivity Analysis}
\label{tab:sensitivity_empirical}

\begin{minipage}{\linewidth}
\footnotesize
The table reports a sensitivity analysis around a pilot-calibrated reference specification with $c=\log(10)$, holding $T_{\mathrm{pilot}}=20$ and $T_{\max}=200$ fixed. For the threshold comparison, the pilot-calibrated working probabilities $(\hat{\pi}_0,\hat{\pi}_1)$ from the reference specification are held fixed across alternative values of $c$. The main empirical specification discussed in the text uses the more conservative threshold $c=\log(1000)$.\\
\end{minipage}
\begin{tabular}{lcccccc}
\toprule
Specification & Model Size & Iterations & $\hat{\pi}_0$ & $\hat{\pi}_{0,\mathrm{raw}}$ & $\hat{\pi}_1$ & $q^\ast$ \\
\midrule
Baseline ($c=\log(10)$, pilot $\pi_0$) & 36 & 7  & 0.102 & 0.102 & 0.974 & 0.611 \\
$c=\log(3)$                            & 34 & 6  & 0.102 & 0.102 & 0.974 & 0.611 \\
$c=\log(30)$                           & 37 & 12 & 0.102 & 0.102 & 0.974 & 0.611 \\
Permutation $\pi_0$ ($c=\log(10)$)     & 27 & 11 & 0.457 & 0.592 & 0.985 & 0.823 \\
\bottomrule
\end{tabular}
\end{table}

The empirical findings are qualitatively stable across reasonable changes in the evidence threshold and calibration method. Varying $c$ changes the selected model size only modestly, while larger thresholds mechanically require more iterations. Permutation calibration yields a noticeably more conservative specification because it implies a much larger null detection probability. Additional discussion is provided in Appendix~\ref{app:empirical_details}.

\section{Concluding Remarks}
\label{sec:conclusion}
This section concludes by summarizing the main contributions of the paper, discussing the limitations of the proposed procedure, and outlining promising directions for future research.

This paper proposes a sequential evidence-aggregation procedure for repeated stochastic variable-selection outcomes generated by bootstrap resampling and stochastic imputation, using PDS as the within-iteration screening device. By modeling detection outcomes across perturbation iterations with a working Bernoulli detection model and accumulating likelihood-ratio–based evidence, the procedure provides a structured inclusion rule together with a data-driven stopping criterion.

The Monte Carlo study evaluates the procedure across 126 simulation
scenarios varying sample size, signal strength, missing-data
mechanism, and missingness level. The results indicate that the
proposed method achieves a favorable trade-off between true positive
and false positive rates relative to the union rule and
frequency-based aggregation, while producing sparser models and
competitive treatment effect estimates. The sequential stopping rule
substantially reduces the number of perturbation iterations compared
to fixed-iteration procedures.

Several limitations should be noted. First, the Bernoulli working
model treats detection indicators as conditionally independent across
perturbation iterations, whereas bootstrap resampling and multiple
imputation from the same observed sample induce positive dependence.
The simulation study provides evidence on the practical consequences
of this approximation, but the procedure should not be interpreted as
delivering exact Bayesian posterior probabilities. Second, the
calibration of the detection probabilities $(\hat{\pi}_0,\hat{\pi}_1)$
relies on a short pilot phase and heuristic estimators. While the
simulation results suggest that the procedure is robust across the
calibration strategies considered, the sensitivity of the evidence
process to pilot length and calibration method warrants further
investigation. Third, the procedure is evaluated here in a linear
model setting with LASSO-based selection; its performance in
nonlinear or high-dimensional nonparametric settings remains to be
studied.

The evidence framework is not restricted to LASSO-based variable
selection. It applies to any procedure that produces binary detection
indicators with different detection probabilities under relevance and
irrelevance. For methods with standard errors, the detection event
can be defined via a significance threshold. For machine learning
methods without standard errors, detection events can be defined
through selection stability, with calibration based on empirical
detection frequencies or permutation of the outcome variable to
estimate the null detection rate. This generality suggests several
directions for future work, including the application of the
framework to random forests and gradient boosting methods, the
development of adaptive calibration strategies that update
$(\hat{\pi}_0,\hat{\pi}_1)$ during the sequential run, and the
extension to settings with multiple treatment variables or
high-dimensional treatment effect heterogeneity.

An important direction for future work is fuller uncertainty propagation across the entire pipeline, for example by embedding the complete perturbation-selection-estimation procedure in an outer resampling scheme. Such an extension would be particularly relevant for causal applications in which uncertainty from model selection is itself substantively important.

\section*{Acknowledgements}
The authors acknowledge support by the state of Baden-Württemberg through bwHPC and the German Research Foundation (DFG) through grant INST 35/1597-1 FUGG. The authors gratefully acknowledge the support of the HERMES network (Higher Education and Research in Management of European Universities) for a travel grant and for facilitating this study through its network.

\bibliography{library.bib}


\clearpage
\appendix

\section{Additional Methodological Details}
\label{app:methodology}

\subsection{Theoretical Properties of the Evidence Process}
\label{app:theory}

This section collects additional theoretical results for the sequential evidence process introduced in Section~\ref{sec:evidence}. The results are derived under an idealized working model in which detection indicators are independent and identically distributed conditional on the relevance status of each variable. They are intended to clarify the classification, stability, and stopping behavior of the proposed rule under the working model.

\paragraph{Setup.}
Let $\{Z_{jt}\}_{t\ge 1}$ denote the detection indicators for variable $j$, and assume
\[
Z_{jt} \sim \mathrm{Bernoulli}(q_j).
\]
Define the log-evidence process
\[
S_{jt} := \log E_{jt}
=
\sum_{s=1}^t X_{js},
\]
where
\[
X_{js}
=
Z_{js}\log\!\left(\frac{\pi_1}{\pi_0}\right)
+
(1-Z_{js})\log\!\left(\frac{1-\pi_1}{1-\pi_0}\right).
\]

\paragraph{Proof of Proposition~\ref{prop:classification}.}

\begin{proof}
Under the idealized working model,
\[
S_{jt} := \log E_{jt}
=
\sum_{s=1}^t X_{js},
\]
where
\[
X_{js}
=
Z_{js}\log\!\left(\frac{\pi_1}{\pi_0}\right)
+
(1-Z_{js})\log\!\left(\frac{1-\pi_1}{1-\pi_0}\right),
\]
and the \(\{X_{js}\}_{s\ge 1}\) are i.i.d. with finite expectation
\[
\mu_j
=
\mathbb{E}[X_{js}]
=
m(q_j;\pi_0,\pi_1).
\]

By the strong law of large numbers,
\[
\frac{S_{jt}}{t}
\to
\mu_j
\qquad\text{almost surely as } t\to\infty.
\]

If \(q_j > q^\ast(\pi_0,\pi_1)\), then by definition
\[
m(q_j;\pi_0,\pi_1)=\mu_j>0,
\]
so \(S_{jt}\to +\infty\) almost surely. Likewise, if
\(q_j < q^\ast(\pi_0,\pi_1)\), then \(\mu_j<0\), so
\(S_{jt}\to -\infty\) almost surely.

Now fix any classification threshold \(c>0\). If \(\mu_j>0\), then
\[
\Prob(S_{jt}>c)\to 1,
\]
and if \(\mu_j<0\), then
\[
\Prob(S_{jt}<-c)\to 1.
\]
Hence, for any fixed threshold, the probability that the sign of the threshold crossing agrees with the sign of the drift converges to one as \(t\to\infty\). This establishes the stated correct-classification result.
\end{proof}

The next result strengthens this intuition by giving an exponential bound on incorrect boundary crossing under nonzero drift.

\begin{proposition}[Exponential bound for incorrect boundary crossing]
\label{prop:errorbound_app}
If $\mu_j>0$, then there exists $\theta_j>0$ such that
\[
\Prob\!\left(\inf_{t\ge 1} S_{jt}\le -c\right)
\le e^{-\theta_j c}.
\]
If $\mu_j<0$, then
\[
\Prob\!\left(\sup_{t\ge 1} S_{jt}\ge c\right)
\le e^{-\theta_j c}.
\]
\end{proposition}

\begin{proof}
Define the moment generating function
\[
\psi_j(\theta):=\mathbb{E}[e^{-\theta X_{js}}].
\]
Since $\psi_j'(0)=-\mu_j<0$, there exists $\theta_j>0$ such that $\psi_j(\theta_j)<1$.

Define the martingale
\[
M_t := \frac{e^{-\theta_j S_{jt}}}{\psi_j(\theta_j)^t}.
\]
Applying optional stopping to the hitting time $\tau_c = \inf\{t: S_{jt}\le -c\}$ yields
\[
\Prob(\tau_c<\infty)\le e^{-\theta_j c}.
\]
\end{proof}

These results show that, under the working model, the evidence process behaves like a random walk with drift: it diverges to $+\infty$ for relevant variables and to $-\infty$ for irrelevant ones, with exponentially small probability of incorrect boundary crossing.

\begin{proposition}[Approximate expected stopping time under drift]
\label{prop:stoppingtime}
Consider the idealized working model in which, for a given variable $j$,
\[
Z_{jt} \sim \mathrm{Bernoulli}(q_j)
\qquad\text{independently over } t,
\]
and let
\[
S_{jt} := \log E_{jt}
=
\sum_{s=1}^t X_{js},
\]
where
\[
X_{js}
=
Z_{js}\log\!\left(\frac{\pi_1}{\pi_0}\right)
+
(1-Z_{js})\log\!\left(\frac{1-\pi_1}{1-\pi_0}\right).
\]
Let
\[
\mu_j := \mathbb{E}[X_{js}]
=
m(q_j;\pi_0,\pi_1),
\]
and define the first-passage stopping time
\[
\tau_j := \inf\{t\ge 1: |S_{jt}| \ge c\},
\]
for a symmetric threshold $c>0$.

If $\mu_j \neq 0$, then for large thresholds $c$, the expected stopping time satisfies the approximation
\[
\mathbb{E}[\tau_j]
\approx
\frac{c}{|\mu_j|}.
\]
In particular, variables whose effective detection probabilities $q_j$ are farther from the break-even value $q^\ast(\pi_0,\pi_1)$ are expected to be classified more rapidly.
\end{proposition}

\begin{proof}
By construction, $\{S_{jt}\}_{t\ge 1}$ is a random walk with i.i.d. increments $\{X_{js}\}_{s\ge 1}$ and drift
\[
\mu_j = \mathbb{E}[X_{js}]
=
q_j\log\!\left(\frac{\pi_1}{\pi_0}\right)
+
(1-q_j)\log\!\left(\frac{1-\pi_1}{1-\pi_0}\right).
\]
If $\mu_j>0$, then the process drifts upward and typically exits through the upper boundary $+c$. If $\mu_j<0$, it drifts downward and typically exits through the lower boundary $-c$.

By the strong law of large numbers,
\[
\frac{S_{jt}}{t} \to \mu_j
\qquad\text{almost surely as } t\to\infty.
\]
Hence, for large $t$, the sample path is well approximated by the deterministic trajectory
\[
S_{jt} \approx t\,\mu_j.
\]
Under this approximation, the boundary crossing condition $|S_{jt}| \ge c$ becomes
\[
t\,|\mu_j| \approx c,
\]
which yields
\[
\tau_j \approx \frac{c}{|\mu_j|}.
\]
Taking expectations gives the approximation
\[
\mathbb{E}[\tau_j]
\approx
\frac{c}{|\mu_j|}.
\]

Finally, since $\mu_j = m(q_j;\pi_0,\pi_1)$ vanishes at
$q_j = q^\ast(\pi_0,\pi_1)$ and increases in absolute value as $q_j$
moves away from that break-even point, variables with stronger positive
or negative drift are expected to hit a decision boundary sooner.
\end{proof}

Proposition~\ref{prop:stoppingtime} provides a heuristic link between statistical separation and computational cost under the working model. Variables with effective detection probabilities close to the break-even level $q^\ast(\pi_0,\pi_1)$ have small drift and therefore tend to remain in the continuation region longer. By contrast, variables with strongly positive or strongly negative drift are classified more quickly. This helps explain why the proposed sequential procedure can reduce the average number of perturbation iterations when many variables are clearly relevant or clearly irrelevant.

\paragraph{Proof of Proposition~\ref{prop:decision}.}
The result follows from a standard Bayesian decision-theoretic argument combined with the factorization of the likelihood ratio under the working detection model.

\begin{proof}[Proof of Proposition~\ref{prop:decision}]
Let $z_j^T = (Z_{j1}, \dots, Z_{jT})$ denote the observed detection sequence for variable $j$. Consider the decision problem with action space
\[
a_1:\ \text{classify the variable as relevant},
\qquad
a_0:\ \text{classify the variable as irrelevant}.
\]

Let $\ell_1 > 0$ denote the loss of classifying a relevant variable as irrelevant (false exclusion), and $\ell_0 > 0$ the loss of classifying an irrelevant variable as relevant (false inclusion).

The posterior expected loss of action $a_1$ is
\[
\rho(a_1 \mid z_j^T)
=
\ell_0 \, \Pr(H_0 \mid z_j^T),
\]
since a loss is incurred only if the variable is in fact irrelevant. Similarly, the posterior expected loss of action $a_0$ is
\[
\rho(a_0 \mid z_j^T)
=
\ell_1 \, \Pr(H_1 \mid z_j^T),
\]
since a loss is incurred only if the variable is in fact relevant.

The Bayes rule selects $a_1$ if and only if
\[
\rho(a_1 \mid z_j^T) < \rho(a_0 \mid z_j^T),
\]
that is,
\[
\ell_0 \, \Pr(H_0 \mid z_j^T)
<
\ell_1 \, \Pr(H_1 \mid z_j^T).
\]
Rearranging yields
\[
\frac{\Pr(H_1 \mid z_j^T)}{\Pr(H_0 \mid z_j^T)}
>
\frac{\ell_0}{\ell_1}.
\]

By Bayes' theorem,
\[
\frac{\Pr(H_1 \mid z_j^T)}{\Pr(H_0 \mid z_j^T)}
=
\frac{\Prob(z_j^T \mid H_1)}{\Prob(z_j^T \mid H_0)}
\cdot
\frac{\Pr(H_1)}{\Pr(H_0)}.
\]
Under conditional independence across $t$, the likelihood ratio factorizes as
\[
\frac{\Prob(z_j^T \mid H_1)}{\Prob(z_j^T \mid H_0)}
=
\prod_{t=1}^T
\frac{\Prob(Z_{jt} \mid H_1)}{\Prob(Z_{jt} \mid H_0)}
=
E_{jT}.
\]
Hence,
\[
E_{jT} \cdot \frac{\Pr(H_1)}{\Pr(H_0)}
>
\frac{\ell_0}{\ell_1},
\]
which is equivalent to
\[
E_{jT}
>
\frac{\ell_0 \, \Pr(H_0)}{\ell_1 \, \Pr(H_1)}.
\]
Taking logarithms yields the stated decision rule:
\[
\log E_{jT}
>
\log\!\left(
\frac{\ell_0 \, \Pr(H_0)}{\ell_1 \, \Pr(H_1)}
\right).
\]

The equal-prior case follows immediately by setting $\Pr(H_0)=\Pr(H_1)$.
\end{proof}

\subsection{Calibration Details}
\label{app:calibration}

The working detection probabilities $(\pi_0,\pi_1)$ are calibrated empirically using pilot detection frequencies.

\paragraph{Pilot frequencies.}
For each variable,
\[
\hat{\pi}_{j,\mathrm{pilot}}
=
\frac{1}{T_{\mathrm{pilot}}}
\sum_{t=1}^{T_{\mathrm{pilot}}} Z_{jt}.
\]

\paragraph{Null calibration.}
We estimate $\pi_0$ using variables with low pilot detection frequencies. Let
\[
q_{0.10}\!\left(\{\hat{\pi}_{j,\mathrm{pilot}}\}_{j=1}^p\right)
\]
denote the 10th empirical quantile of the pilot frequencies. We then define
\[
\mathcal{L}_{\mathrm{pilot}}
=
\left\{
j :
\hat{\pi}_{j,\mathrm{pilot}}
\le
q_{0.10}\!\left(\{\hat{\pi}_{j,\mathrm{pilot}}\}_{j=1}^p\right)
\right\},
\qquad
\hat{\pi}_0^{\mathrm{raw}}
=
\frac{1}{|\mathcal{L}_{\mathrm{pilot}}|}
\sum_{j\in\mathcal{L}_{\mathrm{pilot}}}
\hat{\pi}_{j,\mathrm{pilot}}.
\]

\paragraph{Alternative calibration.}
We estimate $\pi_1$ using variables with high pilot detection frequencies. Let
\[
q_{0.50}\!\left(\{\hat{\pi}_{j,\mathrm{pilot}}\}_{j=1}^p\right)
\]
denote the empirical median of the pilot frequencies. We then define
\[
\mathcal{H}_{\mathrm{pilot}}
=
\left\{
j :
\hat{\pi}_{j,\mathrm{pilot}}
>
q_{0.50}\!\left(\{\hat{\pi}_{j,\mathrm{pilot}}\}_{j=1}^p\right)
\right\},
\qquad
\hat{\pi}_1
=
\frac{1}{|\mathcal{H}_{\mathrm{pilot}}|}
\sum_{j\in\mathcal{H}_{\mathrm{pilot}}}
\hat{\pi}_{j,\mathrm{pilot}}.
\]

\paragraph{Stabilization.}
To avoid overly aggressive evidence accumulation, we regularize the null estimate:
\[
\tilde{\pi}_0
=
\max\!\left\{
\pi_{0,\min},
(1-\lambda_0)\hat{\pi}_0^{\mathrm{raw}} + \lambda_0 \alpha
\right\}.
\]

We additionally enforce a minimum break-even detection rate $q^\ast_{\min}$ by increasing $\tilde{\pi}_0$ if necessary.

\paragraph{Permutation calibration.}
Permutation-based calibration breaks either the outcome or treatment relationship but not both simultaneously under the asymmetric detection rule. As a result, it yields upward-biased estimates of $\pi_0$ and is therefore used only as a robustness check.

\paragraph{Fallback rule.}
If $\hat{\pi}_1 \le \hat{\pi}_0$, the pilot phase is extended. If separation is still not achieved, the procedure falls back to frequency thresholding.

\subsection{Additional Remarks on the Working Model}
\label{app:workingmodel}

The Bernoulli detection model assumes conditional independence across perturbation iterations. In practice, bootstrap resampling and stochastic imputation induce dependence. The model should therefore be interpreted as a parsimonious approximation to marginal detection behavior.

The resulting evidence measure is not an exact Bayes factor but a likelihood-ratio–type score derived from a working model. Its practical validity is evaluated empirically in the simulation study.

\subsection{Additional Implementation Details}
\label{app:implementation}

The procedure uses one stochastic imputation per perturbation iteration. This avoids the need for within-iteration aggregation and keeps the detection process binary. Rubin pooling is applied only after variable selection.

Final standard errors reflect estimation and imputation uncertainty conditional on the selected model but do not fully account for selection uncertainty.

\section{Additional Simulation Results}
\label{app:sim_tables}

This appendix reports additional disaggregated results from the Monte Carlo study described in Section~\ref{sec:simulation}. The results complement the aggregate findings in the main text and show that the qualitative ranking of methods is stable across scenario dimensions.

\subsection{Fixed-Budget Selection Performance by Simulation Dimension}

Tables in this subsection report fixed-budget selection results by sample size, missingness rate, missing-data mechanism, signal strength, error structure, and design. Across these dimensions, the proposed method consistently maintains a higher true positive rate than the 50\% frequency threshold, while the latter remains more conservative in terms of false positive rate and model size.

\paragraph{Sample size.}
Selection performance improves with sample size for all methods. The proposed method's TPR rises from 0.653 at $n=100$ to 0.885 at $n=1000$,
while the 50\% threshold rises from 0.486 to 0.775. The TPR advantage of the proposed method therefore persists across all sample sizes, although it comes with a higher FPR.

\paragraph{Missingness rate.}
The proposed method's TPR increases from 0.728 at 20\% missingness to 0.805 at 60\% missingness, while the 50\% threshold improves from 0.571 to 0.676. The qualitative ordering of methods is therefore stable across missingness levels.

\paragraph{Missing-data mechanism.}
The largest differences appear under MNAR, where all methods perform worse. Under MNAR, the proposed method attains a TPR of 0.762 compared with 0.602 for the 50\% threshold. Under MCAR and MAR, the gap is smaller but remains substantial.

\paragraph{Signal strength and error structure.}
Higher signal strength improves selection performance for all methods, while heteroscedasticity modestly reduces TPR. These changes do not alter the overall ranking of methods.

\begin{table}[t]
\caption{Fixed-Budget Variable Selection Performance by $n$}
\label{tab:fixed_by_nn}
\centering
\begin{tabular}{llcccc}
\toprule
$n$ & Method & TPR & FPR & Dist. to Ideal & Model Size \\
\midrule
100 & Union Rule & 1.000 & 1.000 & 1.000 & 50.0 \\
100 & Frequency threshold (50\%) & 0.486 & 0.065 & 0.522 & 5.4 \\
100 & Frequency threshold (75\%) & 0.302 & 0.006 & 0.698 & 1.8 \\
100 & Proposed Method & 0.653 & 0.299 & 0.487 & 16.7 \\
500 & Union Rule & 1.000 & 1.000 & 1.000 & 50.0 \\
500 & Frequency threshold (50\%) & 0.737 & 0.084 & 0.291 & 7.5 \\
500 & Frequency threshold (75\%) & 0.565 & 0.007 & 0.435 & 3.1 \\
500 & Proposed Method & 0.857 & 0.437 & 0.479 & 23.9 \\
1000 & Union Rule & 1.000 & 1.000 & 1.000 & 50.0 \\
1000 & Frequency threshold (50\%) & 0.775 & 0.086 & 0.259 & 7.7 \\
1000 & Frequency threshold (75\%) & 0.603 & 0.006 & 0.398 & 3.3 \\
1000 & Proposed Method & 0.885 & 0.474 & 0.504 & 25.7 \\
\bottomrule
\end{tabular}
\end{table}

\begin{table}[t]
\caption{Fixed-Budget Variable Selection Performance by Missing rate}
\label{tab:fixed_by_miss_rate}
\centering
\begin{tabular}{llcccc}
\toprule
Missing rate & Method & TPR & FPR & Dist. to Ideal & Model Size \\
\midrule
20\% & Union Rule & 1.000 & 1.000 & 1.000 & 50.0 \\
20\% & Frequency threshold (50\%) & 0.571 & 0.076 & 0.446 & 6.3 \\
20\% & Frequency threshold (75\%) & 0.395 & 0.007 & 0.606 & 2.3 \\
20\% & Proposed Method & 0.728 & 0.373 & 0.508 & 20.4 \\
40\% & Union Rule & 1.000 & 1.000 & 1.000 & 50.0 \\
40\% & Frequency threshold (50\%) & 0.647 & 0.080 & 0.375 & 6.8 \\
40\% & Frequency threshold (75\%) & 0.469 & 0.006 & 0.532 & 2.6 \\
40\% & Proposed Method & 0.780 & 0.386 & 0.482 & 21.3 \\
60\% & Union Rule & 1.000 & 1.000 & 1.000 & 50.0 \\
60\% & Frequency threshold (50\%) & 0.676 & 0.073 & 0.346 & 6.7 \\
60\% & Frequency threshold (75\%) & 0.500 & 0.005 & 0.501 & 2.7 \\
60\% & Proposed Method & 0.805 & 0.390 & 0.472 & 21.6 \\
\bottomrule
\end{tabular}
\end{table}

\begin{table}[t]
\caption{Fixed-Budget Variable Selection Performance by Missingness}
\label{tab:fixed_by_miss_type}
\centering
\begin{tabular}{llcccc}
\toprule
Missingness & Method & TPR & FPR & Dist. to Ideal & Model Size \\
\midrule
MCAR & Union Rule & 1.000 & 1.000 & 1.000 & 50.0 \\
MCAR & Frequency threshold (50\%) & 0.670 & 0.071 & 0.348 & 6.5 \\
MCAR & Frequency threshold (75\%) & 0.500 & 0.006 & 0.500 & 2.8 \\
MCAR & Proposed Method & 0.790 & 0.352 & 0.445 & 19.8 \\
MAR & Union Rule & 1.000 & 1.000 & 1.000 & 50.0 \\
MAR & Frequency threshold (50\%) & 0.659 & 0.072 & 0.360 & 6.5 \\
MAR & Frequency threshold (75\%) & 0.483 & 0.006 & 0.517 & 2.7 \\
MAR & Proposed Method & 0.780 & 0.352 & 0.450 & 19.7 \\
MNAR & Union Rule & 1.000 & 1.000 & 1.000 & 50.0 \\
MNAR & Frequency threshold (50\%) & 0.602 & 0.082 & 0.421 & 6.7 \\
MNAR & Frequency threshold (75\%) & 0.420 & 0.006 & 0.581 & 2.4 \\
MNAR & Proposed Method & 0.762 & 0.420 & 0.530 & 22.7 \\
\bottomrule
\end{tabular}
\end{table}

\begin{table}[t]
\caption{Fixed-Budget Variable Selection Performance by $R^2$}
\label{tab:fixed_by_target.R2}
\centering
\begin{tabular}{llcccc}
\toprule
$R^2$ & Method & TPR & FPR & Dist. to Ideal & Model Size \\
\midrule
R² = 0.2 & Union Rule & 1.000 & 1.000 & 1.000 & 50.0 \\
R² = 0.2 & Frequency threshold (50\%) & 0.571 & 0.089 & 0.450 & 6.9 \\
R² = 0.2 & Frequency threshold (75\%) & 0.381 & 0.008 & 0.619 & 2.3 \\
R² = 0.2 & Proposed Method & 0.730 & 0.404 & 0.529 & 21.8 \\
R² = 0.6 & Union Rule & 1.000 & 1.000 & 1.000 & 50.0 \\
R² = 0.6 & Frequency threshold (50\%) & 0.702 & 0.063 & 0.318 & 6.3 \\
R² = 0.6 & Frequency threshold (75\%) & 0.538 & 0.004 & 0.463 & 2.9 \\
R² = 0.6 & Proposed Method & 0.819 & 0.363 & 0.441 & 20.4 \\
\bottomrule
\end{tabular}
\end{table}

\begin{table}[t]
\caption{Fixed-Budget Variable Selection Performance by Error structure}
\label{tab:fixed_by_heteroscedastic}
\centering
\begin{tabular}{llcccc}
\toprule
Error structure & Method & TPR & FPR & Dist. to Ideal & Model Size \\
\midrule
Homoscedastic & Union Rule & 1.000 & 1.000 & 1.000 & 50.0 \\
Homoscedastic & Frequency threshold (50\%) & 0.657 & 0.076 & 0.365 & 6.7 \\
Homoscedastic & Frequency threshold (75\%) & 0.481 & 0.006 & 0.519 & 2.7 \\
Homoscedastic & Proposed Method & 0.790 & 0.390 & 0.482 & 21.5 \\
Heteroscedastic & Union Rule & 1.000 & 1.000 & 1.000 & 50.0 \\
Heteroscedastic & Frequency threshold (50\%) & 0.609 & 0.076 & 0.410 & 6.5 \\
Heteroscedastic & Frequency threshold (75\%) & 0.430 & 0.006 & 0.570 & 2.4 \\
Heteroscedastic & Proposed Method & 0.754 & 0.375 & 0.491 & 20.6 \\
\bottomrule
\end{tabular}
\end{table}

\begin{table}[t]
\caption{Fixed-Budget Variable Selection Performance by Design}
\label{tab:fixed_by_design}
\centering
\begin{tabular}{llcccc}
\toprule
Design & Method & TPR & FPR & Dist. to Ideal & Model Size \\
\midrule
4 & Union Rule & 1.000 & 1.000 & 1.000 & 50.0 \\
4 & Frequency threshold (50\%) & 0.618 & 0.075 & 0.401 & 6.5 \\
4 & Frequency threshold (75\%) & 0.440 & 0.006 & 0.561 & 2.5 \\
4 & Proposed Method & 0.760 & 0.374 & 0.487 & 20.6 \\
5 & Union Rule & 1.000 & 1.000 & 1.000 & 50.0 \\
5 & Frequency threshold (50\%) & 0.748 & 0.084 & 0.280 & 7.5 \\
5 & Frequency threshold (75\%) & 0.578 & 0.006 & 0.422 & 3.2 \\
5 & Proposed Method & 0.865 & 0.443 & 0.481 & 24.2 \\
\bottomrule
\end{tabular}
\end{table}

\FloatBarrier

\subsection{Matched-Budget Selection Performance}

The matched-budget results show that the qualitative conclusions from the fixed-budget comparison remain unchanged when benchmark methods are restricted to the same average number of perturbation iterations as the proposed stopping rule. This indicates that the observed performance differences are not driven solely by computational budget.

\begin{table}[t]
\caption{Matched-Budget Variable Selection Performance by $n$}
\label{tab:matched_by_nn}
\centering
\begin{tabular}{llcccc}
\toprule
$n$ & Method & TPR & FPR & Dist. to Ideal & Model Size \\
\midrule
100 & Union Rule (matched) & 1.000 & 1.000 & 1.000 & 50.0 \\
100 & Frequency threshold (50\%, matched) & 0.487 & 0.069 & 0.522 & 5.5 \\
100 & Frequency threshold (75\%, matched) & 0.305 & 0.006 & 0.696 & 1.8 \\
100 & Proposed Method & 0.653 & 0.299 & 0.487 & 16.7 \\
500 & Union Rule (matched) & 1.000 & 1.000 & 1.000 & 50.0 \\
500 & Frequency threshold (50\%, matched) & 0.736 & 0.093 & 0.296 & 7.9 \\
500 & Frequency threshold (75\%, matched) & 0.566 & 0.009 & 0.434 & 3.2 \\
500 & Proposed Method & 0.857 & 0.437 & 0.479 & 23.9 \\
1000 & Union Rule (matched) & 1.000 & 1.000 & 1.000 & 50.0 \\
1000 & Frequency threshold (50\%, matched) & 0.772 & 0.097 & 0.266 & 8.2 \\
1000 & Frequency threshold (75\%, matched) & 0.605 & 0.008 & 0.396 & 3.4 \\
1000 & Proposed Method & 0.885 & 0.474 & 0.504 & 25.7 \\
\bottomrule
\end{tabular}
\end{table}

\begin{table}[t]
\caption{Matched-Budget Variable Selection Performance by Missing rate}
\label{tab:matched_by_miss_rate}
\centering
\begin{tabular}{llcccc}
\toprule
Missing rate & Method & TPR & FPR & Dist. to Ideal & Model Size \\
\midrule
20\% & Union Rule (matched) & 1.000 & 1.000 & 1.000 & 50.0 \\
20\% & Frequency threshold (50\%, matched) & 0.571 & 0.083 & 0.449 & 6.6 \\
20\% & Frequency threshold (75\%, matched) & 0.397 & 0.009 & 0.603 & 2.4 \\
20\% & Proposed Method & 0.728 & 0.373 & 0.508 & 20.4 \\
40\% & Union Rule (matched) & 1.000 & 1.000 & 1.000 & 50.0 \\
40\% & Frequency threshold (50\%, matched) & 0.646 & 0.087 & 0.378 & 7.2 \\
40\% & Frequency threshold (75\%, matched) & 0.470 & 0.008 & 0.530 & 2.7 \\
40\% & Proposed Method & 0.780 & 0.386 & 0.482 & 21.3 \\
60\% & Union Rule (matched) & 1.000 & 1.000 & 1.000 & 50.0 \\
60\% & Frequency threshold (50\%, matched) & 0.676 & 0.080 & 0.349 & 7.0 \\
60\% & Frequency threshold (75\%, matched) & 0.501 & 0.007 & 0.499 & 2.8 \\
60\% & Proposed Method & 0.805 & 0.390 & 0.472 & 21.6 \\
\bottomrule
\end{tabular}
\end{table}

\begin{table}[t]
\caption{Matched-Budget Variable Selection Performance by Missingness}
\label{tab:matched_by_miss_type}
\centering
\begin{tabular}{llcccc}
\toprule
Missingness & Method & TPR & FPR & Dist. to Ideal & Model Size \\
\midrule
MCAR & Union Rule (matched) & 1.000 & 1.000 & 1.000 & 50.0 \\
MCAR & Frequency threshold (50\%, matched) & 0.670 & 0.077 & 0.351 & 6.8 \\
MCAR & Frequency threshold (75\%, matched) & 0.502 & 0.008 & 0.498 & 2.9 \\
MCAR & Proposed Method & 0.790 & 0.352 & 0.445 & 19.8 \\
MAR & Union Rule (matched) & 1.000 & 1.000 & 1.000 & 50.0 \\
MAR & Frequency threshold (50\%, matched) & 0.658 & 0.078 & 0.363 & 6.8 \\
MAR & Frequency threshold (75\%, matched) & 0.485 & 0.007 & 0.515 & 2.8 \\
MAR & Proposed Method & 0.780 & 0.352 & 0.450 & 19.7 \\
MNAR & Union Rule (matched) & 1.000 & 1.000 & 1.000 & 50.0 \\
MNAR & Frequency threshold (50\%, matched) & 0.602 & 0.091 & 0.424 & 7.1 \\
MNAR & Frequency threshold (75\%, matched) & 0.422 & 0.008 & 0.579 & 2.5 \\
MNAR & Proposed Method & 0.762 & 0.420 & 0.530 & 22.7 \\
\bottomrule
\end{tabular}
\end{table}

\begin{table}[t]
\caption{Matched-Budget Variable Selection Performance by $R^2$}
\label{tab:matched_by_target.R2}
\centering
\begin{tabular}{llcccc}
\toprule
$R^2$ & Method & TPR & FPR & Dist. to Ideal & Model Size \\
\midrule
R² = 0.2 & Union Rule (matched) & 1.000 & 1.000 & 1.000 & 50.0 \\
R² = 0.2 & Frequency threshold (50\%, matched) & 0.571 & 0.097 & 0.452 & 7.2 \\
R² = 0.2 & Frequency threshold (75\%, matched) & 0.384 & 0.010 & 0.617 & 2.4 \\
R² = 0.2 & Proposed Method & 0.730 & 0.404 & 0.529 & 21.8 \\
R² = 0.6 & Union Rule (matched) & 1.000 & 1.000 & 1.000 & 50.0 \\
R² = 0.6 & Frequency threshold (50\%, matched) & 0.701 & 0.070 & 0.322 & 6.6 \\
R² = 0.6 & Frequency threshold (75\%, matched) & 0.539 & 0.006 & 0.462 & 2.9 \\
R² = 0.6 & Proposed Method & 0.819 & 0.363 & 0.441 & 20.4 \\
\bottomrule
\end{tabular}
\end{table}

\begin{table}[t]
\caption{Matched-Budget Variable Selection Performance by Error structure}
\label{tab:matched_by_heteroscedastic}
\centering
\begin{tabular}{llcccc}
\toprule
Error structure & Method & TPR & FPR & Dist. to Ideal & Model Size \\
\midrule
Homoscedastic & Union Rule (matched) & 1.000 & 1.000 & 1.000 & 50.0 \\
Homoscedastic & Frequency threshold (50\%, matched) & 0.656 & 0.084 & 0.368 & 7.1 \\
Homoscedastic & Frequency threshold (75\%, matched) & 0.483 & 0.008 & 0.517 & 2.8 \\
Homoscedastic & Proposed Method & 0.790 & 0.390 & 0.482 & 21.5 \\
Heteroscedastic & Union Rule (matched) & 1.000 & 1.000 & 1.000 & 50.0 \\
Heteroscedastic & Frequency threshold (50\%, matched) & 0.609 & 0.083 & 0.412 & 6.8 \\
Heteroscedastic & Frequency threshold (75\%, matched) & 0.432 & 0.008 & 0.568 & 2.5 \\
Heteroscedastic & Proposed Method & 0.754 & 0.375 & 0.491 & 20.6 \\
\bottomrule
\end{tabular}
\end{table}

\begin{table}[t]
\caption{Matched-Budget Variable Selection Performance by Design}
\label{tab:matched_by_design}
\centering
\begin{tabular}{llcccc}
\toprule
Design & Method & TPR & FPR & Dist. to Ideal & Model Size \\
\midrule
4 & Union Rule (matched) & 1.000 & 1.000 & 1.000 & 50.0 \\
4 & Frequency threshold (50\%, matched) & 0.618 & 0.082 & 0.404 & 6.8 \\
4 & Frequency threshold (75\%, matched) & 0.442 & 0.008 & 0.559 & 2.5 \\
4 & Proposed Method & 0.760 & 0.374 & 0.487 & 20.6 \\
5 & Union Rule (matched) & 1.000 & 1.000 & 1.000 & 50.0 \\
5 & Frequency threshold (50\%, matched) & 0.749 & 0.094 & 0.284 & 8.0 \\
5 & Frequency threshold (75\%, matched) & 0.579 & 0.008 & 0.421 & 3.3 \\
5 & Proposed Method & 0.865 & 0.443 & 0.481 & 24.2 \\
\bottomrule
\end{tabular}
\end{table}

\FloatBarrier

\subsection{Fixed-Budget Treatment Effect Estimation}

The following tables report treatment-effect performance by scenario dimension under the fixed-budget comparison. The main pattern mirrors the aggregate findings in the main text: the proposed method achieves estimation performance similar to the 50\% threshold, while the 75\% threshold tends to incur larger bias and lower coverage because of under-selection.

\begin{table}[t]
\caption{Fixed-Budget Treatment Effect Estimation by $n$}
\label{tab:treatment_fixed_by_nn}
\centering
\begin{tabular}{llcccc}
\toprule
$n$ & Method & Bias & RMSE & Coverage & Model Size \\
\midrule
100 & Union Rule & -0.043 & 0.218 & 0.988 & 50.0 \\
100 & Frequency threshold (50\%) & -0.017 & 0.208 & 0.946 & 5.4 \\
100 & Frequency threshold (75\%) & 0.035 & 0.245 & 0.911 & 1.8 \\
100 & Proposed Method & -0.029 & 0.205 & 0.962 & 16.7 \\
500 & Union Rule & 0.041 & 0.164 & 0.820 & 50.0 \\
500 & Frequency threshold (50\%) & 0.053 & 0.157 & 0.814 & 7.5 \\
500 & Frequency threshold (75\%) & 0.057 & 0.162 & 0.802 & 3.1 \\
500 & Proposed Method & 0.051 & 0.161 & 0.812 & 23.9 \\
1000 & Union Rule & 0.088 & 0.182 & 0.642 & 50.0 \\
1000 & Frequency threshold (50\%) & 0.096 & 0.176 & 0.636 & 7.7 \\
1000 & Frequency threshold (75\%) & 0.094 & 0.182 & 0.627 & 3.3 \\
1000 & Proposed Method & 0.096 & 0.181 & 0.636 & 25.7 \\
\bottomrule
\end{tabular}
\end{table}

\begin{table}[t]
\caption{Fixed-Budget Treatment Effect Estimation by Missing rate}
\label{tab:treatment_fixed_by_miss_rate}
\centering
\begin{tabular}{llcccc}
\toprule
Missing rate & Method & Bias & RMSE & Coverage & Model Size \\
\midrule
20\% & Union Rule & -0.002 & 0.230 & 0.860 & 50.0 \\
20\% & Frequency threshold (50\%) & 0.037 & 0.227 & 0.828 & 6.3 \\
20\% & Frequency threshold (75\%) & 0.071 & 0.269 & 0.794 & 2.3 \\
20\% & Proposed Method & 0.017 & 0.223 & 0.842 & 20.4 \\
40\% & Union Rule & 0.033 & 0.187 & 0.851 & 50.0 \\
40\% & Frequency threshold (50\%) & 0.043 & 0.178 & 0.828 & 6.8 \\
40\% & Frequency threshold (75\%) & 0.061 & 0.189 & 0.809 & 2.6 \\
40\% & Proposed Method & 0.043 & 0.181 & 0.835 & 21.3 \\
60\% & Union Rule & 0.006 & 0.159 & 0.884 & 50.0 \\
60\% & Frequency threshold (50\%) & 0.012 & 0.145 & 0.872 & 6.7 \\
60\% & Frequency threshold (75\%) & 0.032 & 0.149 & 0.860 & 2.7 \\
60\% & Proposed Method & 0.012 & 0.150 & 0.871 & 21.6 \\
\bottomrule
\end{tabular}
\end{table}

\begin{table}[t]
\caption{Fixed-Budget Treatment Effect Estimation by Missingness}
\label{tab:treatment_fixed_by_miss_type}
\centering
\begin{tabular}{llcccc}
\toprule
Missingness & Method & Bias & RMSE & Coverage & Model Size \\
\midrule
MCAR & Union Rule & -0.022 & 0.136 & 0.945 & 50.0 \\
MCAR & Frequency threshold (50\%) & -0.022 & 0.125 & 0.928 & 6.5 \\
MCAR & Frequency threshold (75\%) & -0.004 & 0.126 & 0.922 & 2.8 \\
MCAR & Proposed Method & -0.020 & 0.127 & 0.936 & 19.8 \\
MAR & Union Rule & -0.026 & 0.139 & 0.949 & 50.0 \\
MAR & Frequency threshold (50\%) & -0.022 & 0.128 & 0.931 & 6.5 \\
MAR & Frequency threshold (75\%) & -0.003 & 0.129 & 0.922 & 2.7 \\
MAR & Proposed Method & -0.021 & 0.128 & 0.939 & 19.7 \\
MNAR & Union Rule & 0.054 & 0.239 & 0.773 & 50.0 \\
MNAR & Frequency threshold (50\%) & 0.089 & 0.232 & 0.747 & 6.7 \\
MNAR & Frequency threshold (75\%) & 0.119 & 0.266 & 0.710 & 2.4 \\
MNAR & Proposed Method & 0.073 & 0.233 & 0.750 & 22.7 \\
\bottomrule
\end{tabular}
\end{table}

\begin{table}[t]
\caption{Fixed-Budget Treatment Effect Estimation by $R^2$}
\label{tab:treatment_fixed_by_target.R2}
\centering
\begin{tabular}{llcccc}
\toprule
$R^2$ & Method & Bias & RMSE & Coverage & Model Size \\
\midrule
R² = 0.2 & Union Rule & -0.072 & 0.174 & 0.928 & 50.0 \\
R² = 0.2 & Frequency threshold (50\%) & -0.044 & 0.161 & 0.909 & 6.9 \\
R² = 0.2 & Frequency threshold (75\%) & -0.022 & 0.167 & 0.903 & 2.3 \\
R² = 0.2 & Proposed Method & -0.057 & 0.164 & 0.917 & 21.8 \\
R² = 0.6 & Union Rule & 0.098 & 0.206 & 0.804 & 50.0 \\
R² = 0.6 & Frequency threshold (50\%) & 0.104 & 0.203 & 0.781 & 6.3 \\
R² = 0.6 & Frequency threshold (75\%) & 0.129 & 0.235 & 0.744 & 2.9 \\
R² = 0.6 & Proposed Method & 0.104 & 0.201 & 0.784 & 20.4 \\
\bottomrule
\end{tabular}
\end{table}

\begin{table}[t]
\caption{Fixed-Budget Treatment Effect Estimation by Error structure}
\label{tab:treatment_fixed_by_heteroscedastic}
\centering
\begin{tabular}{llcccc}
\toprule
Error structure & Method & Bias & RMSE & Coverage & Model Size \\
\midrule
Homoscedastic & Union Rule & 0.017 & 0.184 & 0.851 & 50.0 \\
Homoscedastic & Frequency threshold (50\%) & 0.032 & 0.175 & 0.834 & 6.7 \\
Homoscedastic & Frequency threshold (75\%) & 0.052 & 0.191 & 0.816 & 2.7 \\
Homoscedastic & Proposed Method & 0.027 & 0.177 & 0.836 & 21.5 \\
Heteroscedastic & Union Rule & 0.005 & 0.199 & 0.887 & 50.0 \\
Heteroscedastic & Frequency threshold (50\%) & 0.026 & 0.192 & 0.860 & 6.5 \\
Heteroscedastic & Frequency threshold (75\%) & 0.053 & 0.218 & 0.836 & 2.4 \\
Heteroscedastic & Proposed Method & 0.018 & 0.191 & 0.870 & 20.6 \\
\bottomrule
\end{tabular}
\end{table}

\begin{table}[t]
\caption{Fixed-Budget Treatment Effect Estimation by Design}
\label{tab:treatment_fixed_by_design}
\centering
\begin{tabular}{llcccc}
\toprule
Design & Method & Bias & RMSE & Coverage & Model Size \\
\midrule
4 & Union Rule & 0.004 & 0.192 & 0.880 & 50.0 \\
4 & Frequency threshold (50\%) & 0.022 & 0.184 & 0.856 & 6.5 \\
4 & Frequency threshold (75\%) & 0.049 & 0.207 & 0.833 & 2.5 \\
4 & Proposed Method & 0.015 & 0.184 & 0.865 & 20.6 \\
5 & Union Rule & 0.063 & 0.182 & 0.782 & 50.0 \\
5 & Frequency threshold (50\%) & 0.073 & 0.175 & 0.779 & 7.5 \\
5 & Frequency threshold (75\%) & 0.075 & 0.179 & 0.769 & 3.2 \\
5 & Proposed Method & 0.072 & 0.180 & 0.766 & 24.2 \\
\bottomrule
\end{tabular}
\end{table}

\FloatBarrier

\subsection{Matched-Budget Treatment Effect Estimation}

Matched-budget treatment-effect results are reported below. As in the main text, the qualitative comparison across methods remains largely unchanged once benchmark methods are evaluated at the same effective iteration budget as the proposed method.

\begin{table}[t]
\caption{Matched-Budget Treatment Effect Estimation by $n$}
\label{tab:treatment_matched_by_nn}
\centering
\begin{tabular}{llcccc}
\toprule
$n$ & Method & Bias & RMSE & Coverage & Model Size \\
\midrule
100 & Union Rule (matched) & -0.042 & 0.218 & 0.989 & 50.0 \\
100 & Frequency threshold (50\%, matched) & -0.017 & 0.207 & 0.946 & 5.5 \\
100 & Frequency threshold (75\%, matched) & 0.034 & 0.243 & 0.910 & 1.8 \\
100 & Proposed Method & -0.029 & 0.205 & 0.962 & 16.7 \\
500 & Union Rule (matched) & 0.042 & 0.164 & 0.818 & 50.0 \\
500 & Frequency threshold (50\%, matched) & 0.054 & 0.157 & 0.814 & 7.9 \\
500 & Frequency threshold (75\%, matched) & 0.058 & 0.162 & 0.805 & 3.2 \\
500 & Proposed Method & 0.051 & 0.161 & 0.812 & 23.9 \\
1000 & Union Rule (matched) & 0.087 & 0.181 & 0.646 & 50.0 \\
1000 & Frequency threshold (50\%, matched) & 0.097 & 0.177 & 0.636 & 8.2 \\
1000 & Frequency threshold (75\%, matched) & 0.095 & 0.182 & 0.627 & 3.4 \\
1000 & Proposed Method & 0.096 & 0.181 & 0.636 & 25.7 \\
\bottomrule
\end{tabular}
\end{table}

\begin{table}[t]
\caption{Matched-Budget Treatment Effect Estimation by Missing rate}
\label{tab:treatment_matched_by_miss_rate}
\centering
\begin{tabular}{llcccc}
\toprule
Missing rate & Method & Bias & RMSE & Coverage & Model Size \\
\midrule
20\% & Union Rule (matched) & -0.002 & 0.231 & 0.860 & 50.0 \\
20\% & Frequency threshold (50\%, matched) & 0.038 & 0.227 & 0.830 & 6.6 \\
20\% & Frequency threshold (75\%, matched) & 0.071 & 0.267 & 0.798 & 2.4 \\
20\% & Proposed Method & 0.017 & 0.223 & 0.842 & 20.4 \\
40\% & Union Rule (matched) & 0.034 & 0.187 & 0.852 & 50.0 \\
40\% & Frequency threshold (50\%, matched) & 0.044 & 0.178 & 0.828 & 7.2 \\
40\% & Frequency threshold (75\%, matched) & 0.062 & 0.189 & 0.809 & 2.7 \\
40\% & Proposed Method & 0.043 & 0.181 & 0.835 & 21.3 \\
60\% & Union Rule (matched) & 0.006 & 0.159 & 0.883 & 50.0 \\
60\% & Frequency threshold (50\%, matched) & 0.013 & 0.145 & 0.871 & 7.0 \\
60\% & Frequency threshold (75\%, matched) & 0.032 & 0.150 & 0.859 & 2.8 \\
60\% & Proposed Method & 0.012 & 0.150 & 0.871 & 21.6 \\
\bottomrule
\end{tabular}
\end{table}

\begin{table}[t]
\caption{Matched-Budget Treatment Effect Estimation by Missingness}
\label{tab:treatment_matched_by_miss_type}
\centering
\begin{tabular}{llcccc}
\toprule
Missingness & Method & Bias & RMSE & Coverage & Model Size \\
\midrule
MCAR & Union Rule (matched) & -0.022 & 0.137 & 0.946 & 50.0 \\
MCAR & Frequency threshold (50\%, matched) & -0.021 & 0.125 & 0.929 & 6.8 \\
MCAR & Frequency threshold (75\%, matched) & -0.004 & 0.126 & 0.921 & 2.9 \\
MCAR & Proposed Method & -0.020 & 0.127 & 0.936 & 19.8 \\
MAR & Union Rule (matched) & -0.026 & 0.138 & 0.947 & 50.0 \\
MAR & Frequency threshold (50\%, matched) & -0.021 & 0.127 & 0.929 & 6.8 \\
MAR & Frequency threshold (75\%, matched) & -0.003 & 0.129 & 0.921 & 2.8 \\
MAR & Proposed Method & -0.021 & 0.128 & 0.939 & 19.7 \\
MNAR & Union Rule (matched) & 0.055 & 0.239 & 0.773 & 50.0 \\
MNAR & Frequency threshold (50\%, matched) & 0.090 & 0.232 & 0.748 & 7.1 \\
MNAR & Frequency threshold (75\%, matched) & 0.119 & 0.265 & 0.713 & 2.5 \\
MNAR & Proposed Method & 0.073 & 0.233 & 0.750 & 22.7 \\
\bottomrule
\end{tabular}
\end{table}

\begin{table}[t]
\caption{Matched-Budget Treatment Effect Estimation by $R^2$}
\label{tab:treatment_matched_by_target.R2}
\centering
\begin{tabular}{llcccc}
\toprule
$R^2$ & Method & Bias & RMSE & Coverage & Model Size \\
\midrule
R² = 0.2 & Union Rule (matched) & -0.072 & 0.174 & 0.928 & 50.0 \\
R² = 0.2 & Frequency threshold (50\%, matched) & -0.043 & 0.161 & 0.909 & 7.2 \\
R² = 0.2 & Frequency threshold (75\%, matched) & -0.022 & 0.167 & 0.903 & 2.4 \\
R² = 0.2 & Proposed Method & -0.057 & 0.164 & 0.917 & 21.8 \\
R² = 0.6 & Union Rule (matched) & 0.098 & 0.207 & 0.804 & 50.0 \\
R² = 0.6 & Frequency threshold (50\%, matched) & 0.105 & 0.203 & 0.781 & 6.6 \\
R² = 0.6 & Frequency threshold (75\%, matched) & 0.129 & 0.234 & 0.746 & 2.9 \\
R² = 0.6 & Proposed Method & 0.104 & 0.201 & 0.784 & 20.4 \\
\bottomrule
\end{tabular}
\end{table}

\begin{table}[t]
\caption{Matched-Budget Treatment Effect Estimation by Error structure}
\label{tab:treatment_matched_by_heteroscedastic}
\centering
\begin{tabular}{llcccc}
\toprule
Error structure & Method & Bias & RMSE & Coverage & Model Size \\
\midrule
Homoscedastic & Union Rule (matched) & 0.017 & 0.185 & 0.850 & 50.0 \\
Homoscedastic & Frequency threshold (50\%, matched) & 0.032 & 0.175 & 0.834 & 7.1 \\
Homoscedastic & Frequency threshold (75\%, matched) & 0.052 & 0.191 & 0.817 & 2.8 \\
Homoscedastic & Proposed Method & 0.027 & 0.177 & 0.836 & 21.5 \\
Heteroscedastic & Union Rule (matched) & 0.006 & 0.199 & 0.888 & 50.0 \\
Heteroscedastic & Frequency threshold (50\%, matched) & 0.027 & 0.192 & 0.859 & 6.8 \\
Heteroscedastic & Frequency threshold (75\%, matched) & 0.053 & 0.217 & 0.836 & 2.5 \\
Heteroscedastic & Proposed Method & 0.018 & 0.191 & 0.870 & 20.6 \\
\bottomrule
\end{tabular}
\end{table}

\begin{table}[t]
\caption{Matched-Budget Treatment Effect Estimation by Design}
\label{tab:treatment_matched_by_design}
\centering
\begin{tabular}{llcccc}
\toprule
Design & Method & Bias & RMSE & Coverage & Model Size \\
\midrule
4 & Union Rule (matched) & 0.004 & 0.192 & 0.881 & 50.0 \\
4 & Frequency threshold (50\%, matched) & 0.023 & 0.184 & 0.856 & 6.8 \\
4 & Frequency threshold (75\%, matched) & 0.049 & 0.207 & 0.834 & 2.5 \\
4 & Proposed Method & 0.015 & 0.184 & 0.865 & 20.6 \\
5 & Union Rule (matched) & 0.063 & 0.182 & 0.778 & 50.0 \\
5 & Frequency threshold (50\%, matched) & 0.073 & 0.175 & 0.781 & 8.0 \\
5 & Frequency threshold (75\%, matched) & 0.076 & 0.179 & 0.770 & 3.3 \\
5 & Proposed Method & 0.072 & 0.180 & 0.766 & 24.2 \\
\bottomrule
\end{tabular}
\end{table}

\FloatBarrier

\subsection{TPR--FPR Frontiers by Sample Size and Missingness Mechanism}

Figures~\ref{fig:frontier_fixed_by_n_misstype} and \ref{fig:frontier_matched_by_n_misstype} report TPR--FPR frontiers by sample size and missingness mechanism. These figures illustrate how the relative positioning of the methods varies across subgroups while preserving the main aggregate trade-off discussed in the text.

\begin{figure}[htbp]
\centering
\caption{TPR--FPR Frontier by Sample Size and Missingness Mechanism (Fixed-Budget)}
\begin{minipage}{0.9\linewidth}
\footnotesize
Each point represents the average TPR and FPR of a method within a
sample-size and missingness-mechanism cell. Rows correspond to
MCAR, MAR, and MNAR; columns to $n\in\{100,500,1000\}$. The cross
marks the ideal point at $(0,1)$.
\end{minipage}
\label{fig:frontier_fixed_by_n_misstype}
\includegraphics[width=\linewidth]{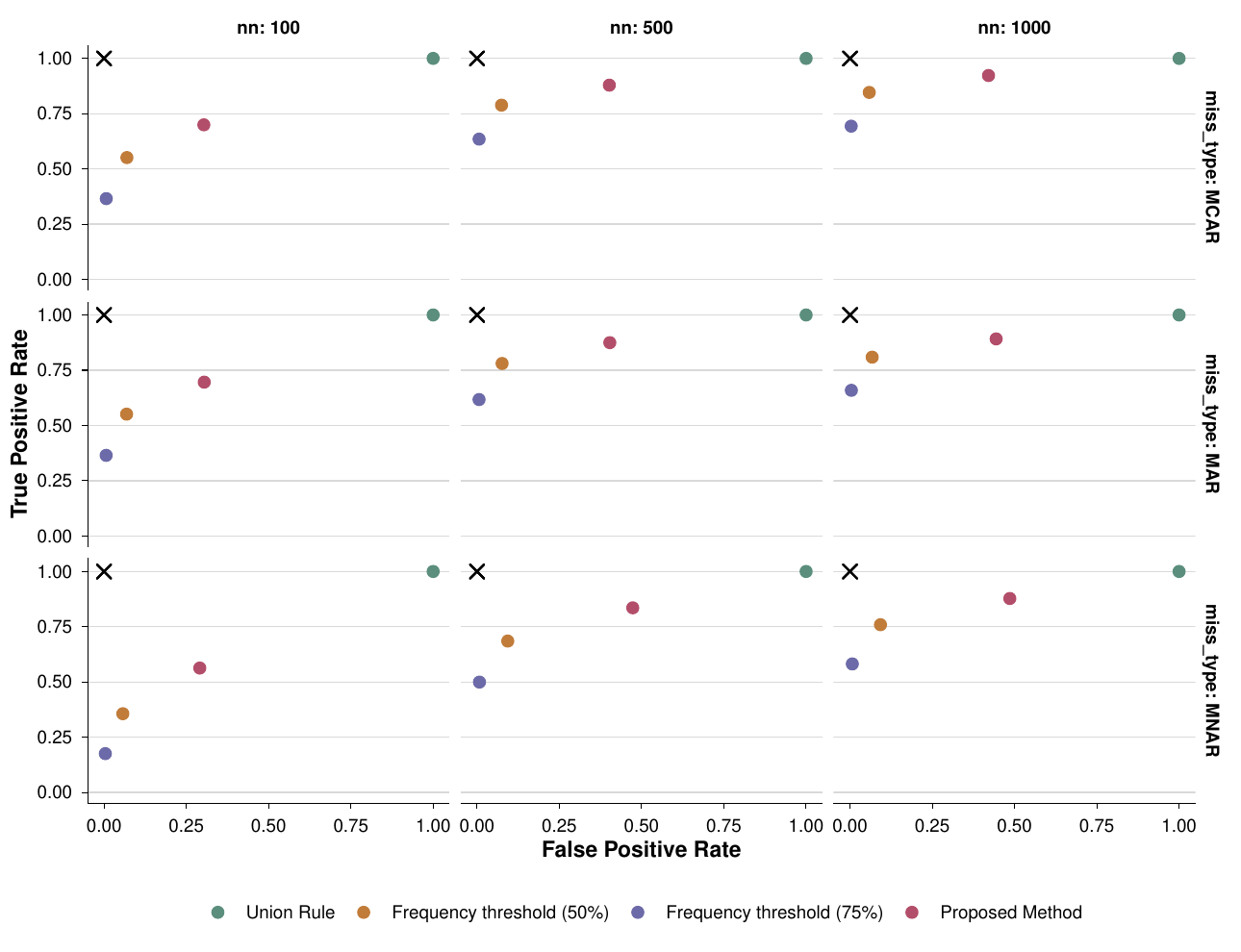}
\end{figure}

\begin{figure}[htbp]
\centering
\caption{TPR--FPR Frontier by Sample Size and Missingness Mechanism (Matched-Budget)}
\begin{minipage}{0.9\linewidth}
\footnotesize
As Figure~\ref{fig:frontier_fixed_by_n_misstype}, but under the
matched-budget comparison.
\end{minipage}
\label{fig:frontier_matched_by_n_misstype}
\includegraphics[width=\linewidth]{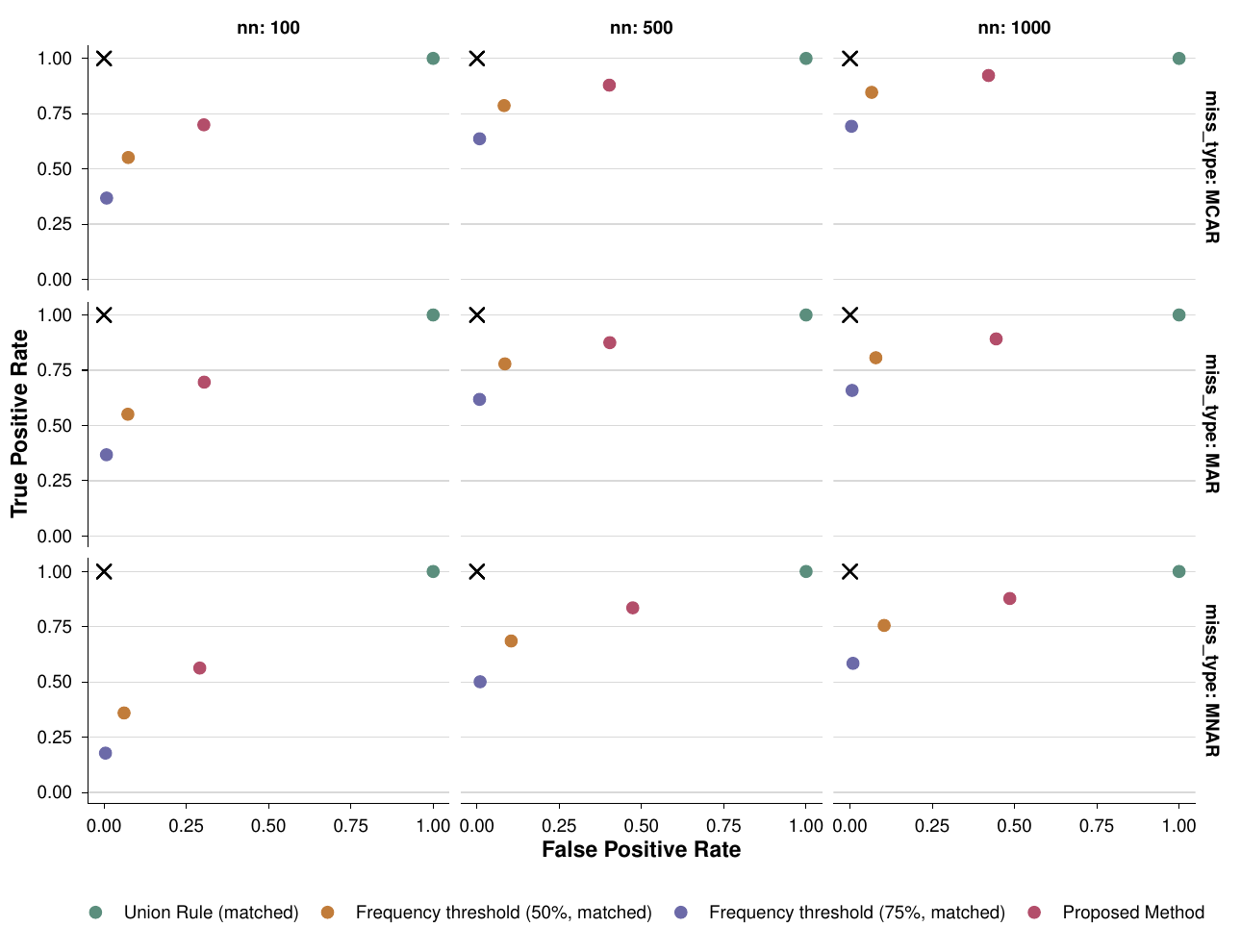}
\end{figure}

\FloatBarrier

\subsection{Distance to Ideal}

Figures~\ref{fig:distance_fixed_by_n}--\ref{fig:distance_matched_by_missrate} report the Euclidean distance to the ideal point $(0,1)$ by subgroup. As discussed in the main text, this metric is included as a descriptive diagnostic, but it should be interpreted cautiously because it weights false positives and false negatives symmetrically.

\begin{figure}[htbp]
\centering
\caption{Mean Distance to Ideal by Sample Size (Fixed-Budget)}
\begin{minipage}{0.9\linewidth}
\footnotesize
Euclidean distance to the ideal point $(0,1)$ by sample size
under the fixed-budget comparison, with $95\%$ confidence intervals.
\end{minipage}
\label{fig:distance_fixed_by_n}
\includegraphics[width=0.75\linewidth]{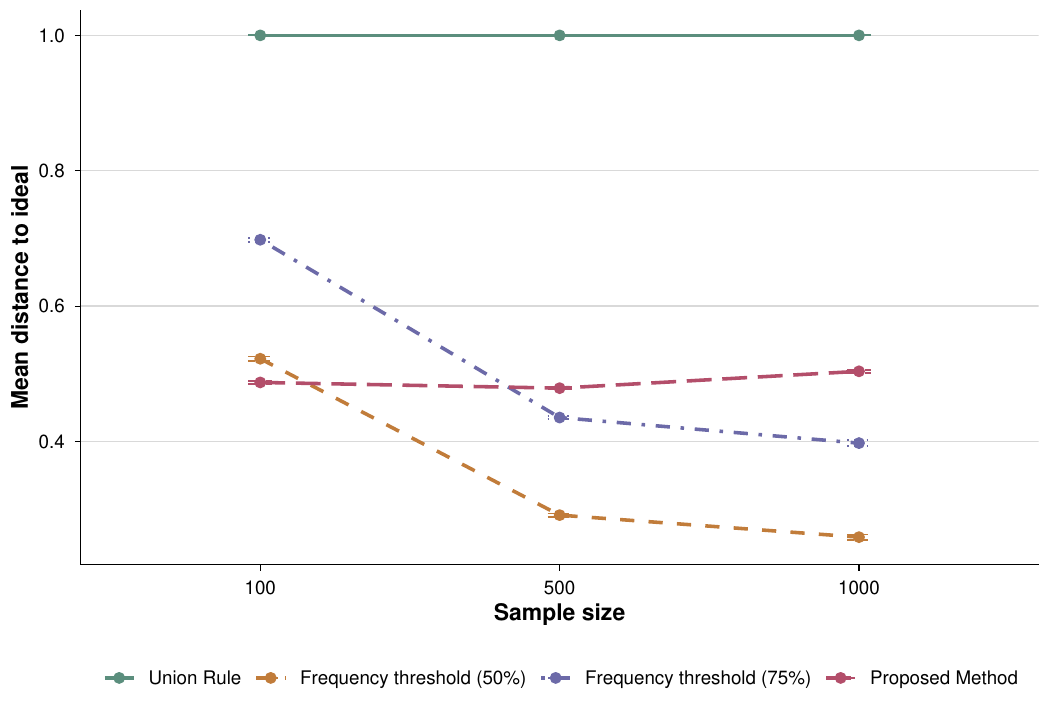}
\end{figure}

\begin{figure}[htbp]
\centering
\caption{Mean Distance to Ideal by Missingness Rate (Fixed-Budget)}
\begin{minipage}{0.9\linewidth}
\footnotesize
Euclidean distance to the ideal point $(0,1)$ by missingness rate
and mechanism under the fixed-budget comparison.
\end{minipage}
\label{fig:distance_fixed_by_missrate}
\includegraphics[width=\linewidth]{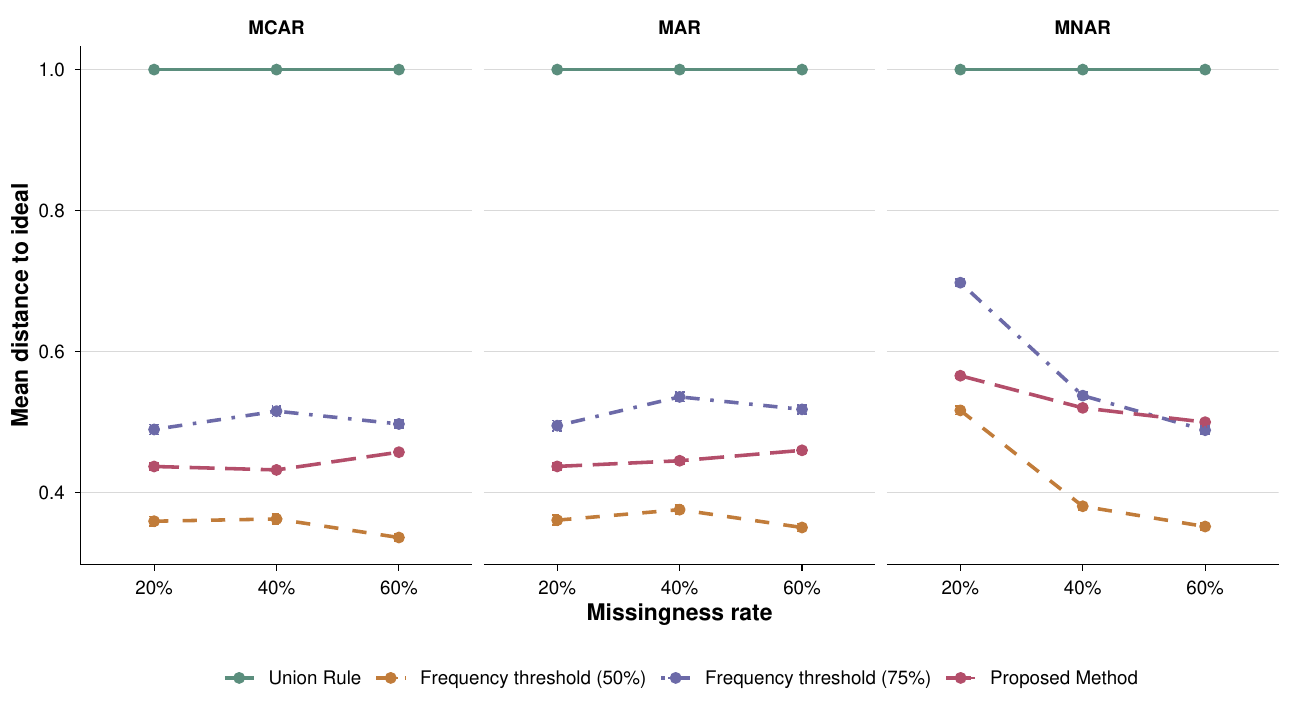}
\end{figure}

\begin{figure}[htbp]
\centering
\caption{Mean Distance to Ideal by Sample Size (Matched-Budget)}
\begin{minipage}{0.9\linewidth}
\footnotesize
As Figure~\ref{fig:distance_fixed_by_n}, but under the matched-budget comparison.
\end{minipage}
\label{fig:distance_matched_by_n}
\includegraphics[width=0.75\linewidth]{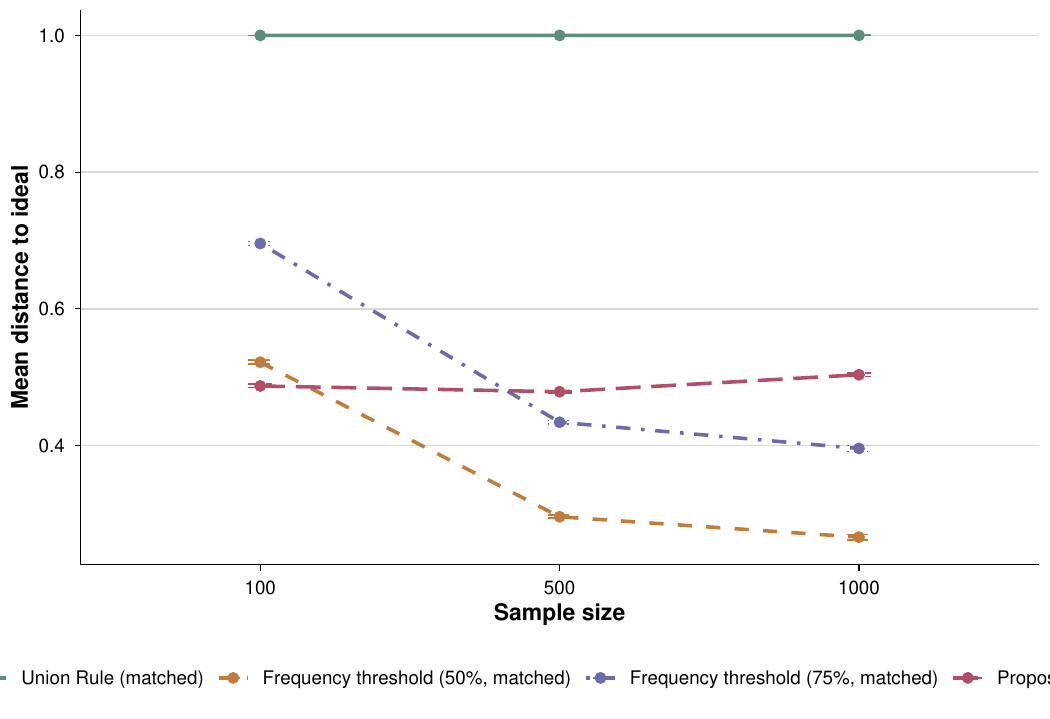}
\end{figure}

\begin{figure}[htbp]
\centering
\caption{Mean Distance to Ideal by Missingness Rate (Matched-Budget)}
\begin{minipage}{0.9\linewidth}
\footnotesize
As Figure~\ref{fig:distance_fixed_by_missrate}, but under the
matched-budget comparison.
\end{minipage}
\label{fig:distance_matched_by_missrate}
\includegraphics[width=\linewidth]{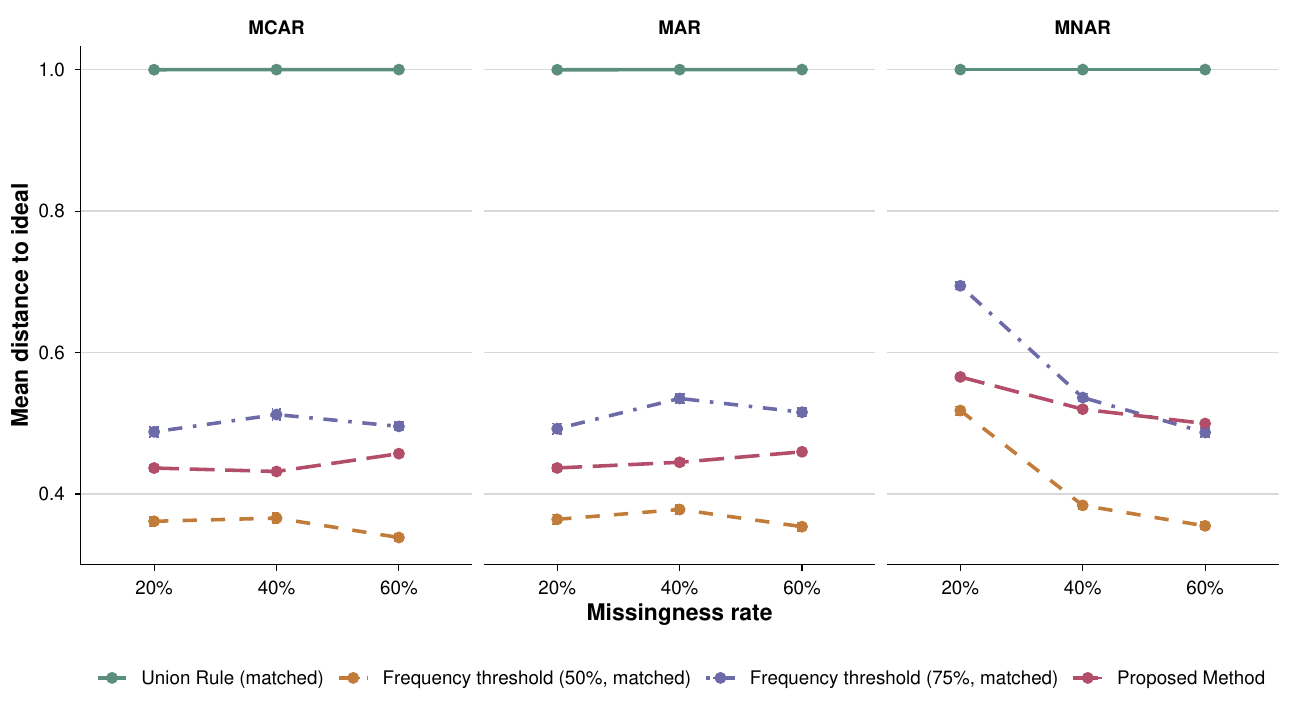}
\end{figure}

\section{Additional Empirical Details}
\label{app:empirical_details}

This appendix provides additional details for the empirical illustration in Section~\ref{sec:empirical}, including variable definitions, missingness rates, a fuller discussion of the estimated coefficient on the variable of interest, and a brief interpretation of the sensitivity analysis.

\subsection{Variable Definitions and Missingness}

Table~\ref{tab:empiricalvars} reports the candidate variables used in the empirical illustration, grouped by category, together with the corresponding ESS variable codes and the share of missing observations in the analysis sample.

\begin{table}[t]
\centering
\caption{Variable Definitions and Missingness}
\label{tab:empiricalvars}
\begin{minipage}{\linewidth}
\footnotesize
The table reports the candidate variables used in the empirical illustration, grouped by category, together with the corresponding ESS variable codes and the share of missing observations in the analysis sample ($n=8{,}543$). Outcome and variable-of-interest measures are fully observed by construction; missing covariate values are handled by MI within the BOOT-MI procedure.\\
\end{minipage}
\begin{tabular}{llr}
\toprule
Category & Variable & Missing (\%) \\
\midrule
Outcome & Employment indicator (pdwrk) & 0.0 \\
Variable of interest & Child aged $\leq 12$ in household (chldo12) & 0.0 \\
\addlinespace
Individual & Age (agea) & 0.0 \\
         & Years of education (eduyrs) & 1.1 \\
         & Education level (eisced) & 0.2 \\
         & Self-rated health (health) & 0.1 \\
\addlinespace
Household & Income decile (hinctnta) & 22.0 \\
         & Subjective financial situation (hincfel) & 1.0 \\
\addlinespace
Family   & Marital status (maritalb) & 0.7 \\
\addlinespace
Attitudes & Happiness (happy) & 0.2 \\
         & Life satisfaction (stflife) & 0.9 \\
         & Generalized trust (ppltrst) & 0.2 \\
         & Trust in parliament (trstprl) & 1.9 \\
         & Trust in legal system (trstlgl) & 1.2 \\
         & Trust in police (trstplc) & 0.5 \\
         & Political interest (polintr) & 0.2 \\
         & Unemployment experience (uempla) & 0.0 \\
\addlinespace
Living environment & Urbanicity (domicil) & 0.1 \\
         & Religious belonging (rlgblg) & 0.8 \\
\addlinespace
Structure & Country indicators (cntry) & 0.0 \\
\bottomrule
\end{tabular}
\end{table}

\subsection{Additional Discussion of the Estimated Coefficient}

The estimated coefficient on the variable of interest is close to zero under all aggregation rules. This pattern suggests that, in the present specification, the estimated conditional association is not materially affected by whether the selected model is relatively sparse or relatively dense.

A few features of the application help explain this result. First, the variable of interest captures the presence of children aged 0--12 in the household rather than focusing on preschool-age children only. Pooling over such a broad age range may attenuate the estimated association because childcare constraints weaken as children age. Second, some controls may partly lie on pathways related to labor supply decisions rather than functioning purely as background covariates. In particular, household income and subjective financial situation may absorb part of the conditional association. Third, the binary employment outcome does not capture intensive-margin adjustments such as reductions in working hours. Fourth, the pooled cross-country specification averages over institutional settings with potentially substantial heterogeneity in childcare provision and labor-market structure.

These considerations do not affect the methodological role of the empirical illustration. The purpose of the application is to compare aggregation rules in a realistic setting rather than to recover a causal parameter.

\subsection{Additional Discussion of the Sensitivity Analysis}

The sensitivity analysis in Table~\ref{tab:sensitivity_empirical} shows that the empirical results are fairly stable across a range of reasonable specifications. When the pilot-calibrated working probabilities are held fixed, varying the threshold from $c=\log(3)$ to $c=\log(30)$ changes the selected model size only modestly, from 34 to 37 variables, while the number of iterations increases from 6 to 12 as the evidence requirement becomes more stringent.

Permutation calibration produces a substantially more conservative model with 27 selected variables. This is consistent with the much larger null detection probability under that calibration, which makes positive detections less informative relative to the working null. Even so, the core set of selected variables remains broadly stable across specifications, supporting the conclusion that the empirical illustration is not driven by a knife-edge calibration choice.

\end{document}